\documentclass[12pt]{amsart}
\usepackage{xcolor}
\usepackage{amssymb}
\usepackage{tikz}
\usepackage{subcaption}
\usepackage{caption}
\usepackage{comment}
\usepackage{pdfpages}
\usepackage{graphicx}
\usepackage{dcolumn}

\RequirePackage[numbers]{natbib}
\RequirePackage[colorlinks=true, pdfstartview=FitV, linkcolor=blue,
  citecolor=blue, urlcolor=blue]{hyperref}
\RequirePackage{hypernat}
\usepackage{paralist}

\usepackage{subcaption}
\graphicspath{{./figures/}}
\usepackage{amsfonts}
\usepackage{amsmath}
\usepackage{amsthm}
\usepackage{amssymb}
\usepackage{amsbsy}
\usepackage{epsfig}
\usepackage{fullpage}
\usepackage{natbib, mathrsfs}

\usepackage{verbatim}
\usepackage[latin1]{inputenc}
\usepackage{mhequ}
\usepackage{algorithm}
\usepackage{algorithmic}
\usepackage[letterpaper,left=1.5in,right=1.5in]{geometry}
\usepackage{enumitem}

\numberwithin{equation}{section}

\def \be{\begin{equs}} 
\def \ee{\end{equs}}

\newtheorem{theorem}{Theorem}[section]
\newtheorem{lemma}[theorem]{Lemma}
\newtheorem{corollary}[theorem]{Corollary}
\newtheorem{prop}[theorem]{Proposition}
\newtheorem*{prop-non}{Proposition}
\newtheorem*{corollary-non}{Corollary}
\newtheorem*{remark-non}{Remark}

\newtheorem{assumption}[theorem]{Assumption}

\theoremstyle{plain}

\newtheorem*{thm-non}{Theorem}

\theoremstyle{definition}

\newtheorem{defn}[theorem]{Definition}
\newtheorem{remark}[theorem]{Remark}

\definecolor{WowColor}{rgb}{.75,0,.75}
\definecolor{SubtleColor}{rgb}{0.9,0,0}

\newcounter{margincounter}

\newcounter{latercounter}

\begin{document}
\allowdisplaybreaks
\title[Accurate and Efficient MCMC for Latent Position Models]
{Accurate and Efficient MCMC for Latent Position Models}

\author{Zonghao Li and Aaron Smith}







\begin{abstract}

Latent position models (LPMs) are a large and popular class of models for random graphs. However, fitting Bayesian LPMs is computationally challenging - computing the likelihood even once takes time that is quadratic in the number of vertices $|V|$ of the observed graph $G = (V,E)$. Many previous papers have introduced approximate MCMC algorithms to speed this up, with the most similar to ours \cite{rastelli2018computationally} presenting an algorithm that has amortized running time that can be reduced almost to $O(|E|)$ and good empirical performance on reasonable inference problems. The present paper offers two algorithms for solving the same problem: a ``fast" algorithm with running time of the same almost-$O(|E|)$ order as \cite{rastelli2018computationally} and much stronger accuracy guarantees, and a ``faster" algorithm with an improved running time of almost $O(|V|)$, and accuracy guarantees that are slightly improved compared to \cite{rastelli2018computationally} (but not sufficient for all tasks). The main improvements come from the introduction of a simple auxiliary data structure that can be cheaply updated during an MCMC run; we suspect that the same ``cheap sketch" may be useful for other MCMC algorithms.
\end{abstract}

\maketitle

\section{Introduction}

Statistical network analysis has emerged as a prominent topic of interest across major fields of study including economics, biology, sociology, and finance (see e.g. textbooks \cite{Wasserman_Faust_1994,Kaminski21}), as well as specialized topics such as education research \cite{sweet2011modeling} and traffic prediction \cite{deng2016latent}. The first statistical network model was the famous (but very simple) Erdos-Renyi random graphs  \cite{erdds1959random}. Almost immediately, statisticians started studying ``geometric" models for networks, in which the nodes were assumed to have unknown ``latent" positions that influenced edge formation \cite{gilbert1961random}. The modern notion of latent position models (LPMs) was introduced by  Hoff et al. \cite{hoff2002latent}, and LPMs and their variants remain popular statistical models for networks (see the survey \cite{kaur2023latent} for many variants and applications).

Many variants of LPMs have been proposed (see \textit{e.g.}    \cite{handcock2007model,krivitsky2009representing,athreya2021estimation} and the survey \cite{kaur2023latent}). In this paper, we focus on the following class of joint models on latent positions and generated graphs:

\begin{defn} \label{DefLPM}
Fix a latent space $\Omega$, prior distribution $\pi_{\Omega}$ on $\Omega$, a symmetric link function $p \, : \, \Omega^{2} \to [0,1]$ satisfying $p(x,y)=p(y,x)$, and size $n$. The \textit{latent position model} associated with the choices $(\Omega, \pi_{\Omega},p,n)$ is a distribution on graphs $G=(V,E)$ and latent positions $Z_{1},\ldots,Z_{n} \in \Omega$, given by setting $V = [n]$, sampling $Z_{1},\ldots,Z_{n} \stackrel{iid}{\sim} \pi_{\Omega}$, and finally sampling the unordered edge set $E_{und} \subseteq \{\{i,j\}:1\leq i<j\leq n\}$ conditional on $Z_{1},\ldots,Z_{n}$ according to the distribution:
\begingroup
\be \label{EqDefLPMLikBasic}
\mathcal{P}[E_{und} | Z_{1},\ldots,Z_{n}] = \prod_{\substack{1\leq i<j\leq n \\  \{i,j\} \in E}} p(Z_{i},Z_{j}) \prod_{\substack{1\leq i<j\leq n\\ \{i,j\} \notin E}} (1- p(Z_{i},Z_{j})).
\ee 
\endgroup
 By a small abuse of notation, we will often call this the latent position model associated with link function $p$ (hiding the dependence on the other choices).
\end{defn}

Unfortunately, it is difficult to fit models of the sort given in Definition \ref{DefLPM} to large graphs using Markov chain Monte Carlo (MCMC). The first difficulty is that the evaluation of the likelihood \eqref{EqDefLPMLikBasic} takes $O(|V|^{2})$ computations, and this likelihood must typically be evaluated many times during an MCMC run. This quadratic complexity has seriously hurt the applicability of Bayesian methods to large datasets. Even this quadratic complexity is often optimistic, as simple Markov chains can mix more slowly for large $|V|$.

Similar problems occur when fitting many Bayesian models to large datasets, and many solutions have been proposed. Our paper's approach is in the category of ``approximate" or ``perturbed" MCMC: we find an MCMC algorithm that is fast, and has a stationary measure that is close \textit{but not equal} to the posterior distribution associated with the model defined in Equation \eqref{EqDefLPMLikBasic}.

There are many ways to ``perturb" nice models such as \eqref{EqDefLPMLikBasic}. We review here caricatures of the most relevant approaches; see \textit{e.g.} \cite{rudolf2024perturbationsmarkovchains} for a recent survey. The most famous and general-purpose approach to perturbing a Markov chain is to use \textit{mini-batch} algorithms such as stochastic gradient Langevin dynamics: at each step of the MCMC algorithm, you take a small random sample of the edges and use them to estimate \eqref{EqDefLPMLikBasic} (see \textit{e.g.} \cite{alquier2016noisy} for an early survey of this approach). These methods are very fast and easy to implement, but they are also known to give poor approximations to the posterior in many cases \cite{nagapetyan2017true,JohndrowJamesE2020NFLf,NEURIPS2018_335cd1b9,johndrow2024freelunchsgld}. Another approach, taken by \cite{rastelli2018computationally} and others, is to chop up the state space $\Omega$ and use a low-resolution proxy for the true likelihood \eqref{EqDefLPMLikBasic}. This is close to the philosophy of the present paper, though unlike \cite{rastelli2018computationally} we don't literally partition $\Omega$. Our main differences come in the details: rather than constructing an explicit hard partition of the state space $\Omega$, we calculate moments associated with an adaptive, implicit and soft partition. These small-looking differences end up having a substantial impact on the asymptotic complexity or accuracy (or both) of our algorithms.

\subsection{Asymptotics and Comparison to  \texorpdfstring{\cite{rastelli2018computationally}}{Rastelli et al., 2024}}\label{SecAsymptoticsComp}

Our paper presents two new algorithms, both based on finding computationally-tractable proxies for the likelihood \eqref{EqDefLPMLikBasic}. We give here a short and informal discussion of the asymptotic running times and error bounds of these algorithms, deferring the technical details and main assumptions to Section \ref{SecAllTheory}. We also compare these bounds to the state-of-the-art results of \cite{rastelli2018computationally}, which takes a fairly similar approach and inspired this paper.

We use two main quantities in comparing the algorithms. The first is what we call the ``amortized cost per sweep." All algorithms considered in this section involve a main loop that changes the embedding of a single node at a time. The ``amortized cost per sweep" is roughly the number of basic computations required to move the latent embedding \textit{of every node}.\footnote{Following standard practice in estimating the complexities of algorithms in computational statistics, we count the steps in our pseudocode descriptions of the algorithms, ignoring the details of \textit{e.g.} how data is stored and retrieved. We also ignore the number of steps required for the chain to mix, which can depend quite strongly on the details of the hyperparameters.} The second quantity is what we call the ``total variation error" - this is just the total variation distance between the posterior distribution associated with the data and the stationary distribution of our MCMC algorithm.

Both our algorithms, and the algorithm in \cite{rastelli2018computationally}, come with a ``size" parameter $b$ that roughly indicates the resolution of the partition of $\Omega$. Our first algorithm, Algorithm \ref{alg_MCMC_full}, also comes with an ``order" parameter $\kappa$, indicating the size of the data structure to maintain.

 The standard Metropolis-within-Gibbs (MwG) algorithm has amortized cost of $O(|V|^{2})$, and the stationary distribution is exactly the posterior distribution.\footnote{A naive implementation of MwG that computes the full likelihood at every step would have running time of $O(|V|^{3})$ per sweep. However, it is straightforward to compute the \textit{change} in likelihood with cost only $O(|V|)$ rather than $O(|V|^{2})$, and we try to measure the complexity of sensible implementations of algorithms rather than the complexity of completely naive implementations.} Algorithm \ref{alg_MCMC_full} has amortized cost of $O(\kappa^{4}|E|+\kappa^{8}|V|b^{-2})$ (where $\kappa$ is a positive constant) 
 and total variation error bounds of the form $O(|V|^{2}(B_g b)^{\kappa +1})$, where $B_g$ is a model-dependent constant independent of $\kappa,n$, and $b$
 (see Theorem \ref{ThemAmRunning} and Section \ref{SecMainErrorBounds} for precise statements of these estimates and the associated assumptions). Algorithm \ref{alg_MCMC_full2} has amortized cost $O(|V| b^{-2})$ and total variation error bounds of the form $O(|V|^{2}(B_g b)^{2})$ (see Theorem \ref{ThemAmRunning2} and Section \ref{SecMainErrorBounds} for precise statements of these estimates and when they hold).

 We note that the technical conditions in our paper are not identical to those in \cite{rastelli2018computationally}. The biggest differences are that they assume that the underlying non-approximate MCMC algorithm mixes quickly (see Assumption 3 of \cite{rastelli2018computationally}), while we assume that there is a good point estimate for the  latent positions $Z$ (see Assumption \ref{assum3}). While these conditions can fail, we expect that \textit{both} of these different technical conditions are satisfied for data sampled from the most common LPMs and the most common underlying MCMC algorithms. In these situations, our bounds compare favorably to the state of the art algorithm \cite{rastelli2018computationally}, which has amortized running time $O(\min(|V| b^{-4} + |E|, |V|^{2}))$ per sweep of all vertices and error bounds of the form $O(b |V|^{2})$ (see Section 6.1 and Theorem 2 of \cite{rastelli2018computationally} for precise statements of these estimates and the associated assumptions). 

We now explain how these bounds relate to each other and those in the abstract. First, the above running time bounds reduce to the simplified bounds in the abstract when $b$ does not depend on $n$. Of course, when $b$ is literally fixed for all $n$ the theoretical error bounds are very weak. Despite this, it is reasonable to run these algorithms in the regime that $b$  goes to 0 slowly with $n$, and the empirical performance can be good in these regimes (see Section \ref{sec-exp} for empirical results and a discussion of which summary statistics converge well even in a regime that the total variation error is large). Allowing $b$ to go to 0 with $n$ doesn't affect our big-picture \textit{comparison} between our algorithms and those of \cite{rastelli2018computationally}. 

Second, the error bounds in Theorem 2 of \cite{rastelli2018computationally} and the application of the bounds in Theorem \ref{Thm1} of this paper to our Algorithm \ref{alg_MCMC_full2} give trivial estimates for most realistic parameter values. To see this, note that our Algorithm \ref{alg_MCMC_full2} only provides a speedup over naive MCMC when $b \gg \frac{1}{\sqrt{|V|}}$, and in this regime our error bound of $O((B_g b)^{2}|V|^{2})$ is still at best order $|V|$, which is larger than 1. The result in \cite{rastelli2018computationally} provides a weaker bound of $O(b |V|^{2})$, so it is also larger than 1 in this regime. In both cases one can use small tricks to slightly improve the running time for very small values of $b$ (see Section 6.1 of \cite{rastelli2018computationally}), but this doesn't affect the overall conclusion: we don't have any nontrivial bound on the error in total variation distances for any realistic parameter values. Despite this, as argued in \cite{rastelli2018computationally} and in Section \ref{sec-exp} of this paper, we can see empirically that the algorithms perform well. 

It is natural to ask if our error bounds are simply very far from sharp. We explore this question in Section \ref{sec-exp}, for a simple model with a Gaussian link function \begin{equation}\label{Gaussian_link}
p(\tau_i,\tau_j; \beta_0,\beta_1,\sigma)
= \beta_0 + \beta_1 \exp\!\left(-\frac{\|\tau_i-\tau_j\|^2}{2\sigma^2}\right).
\end{equation}

This exploration shows that there is a genuinely large error in the likelihood, so the ineffectiveness of our bounds and \cite{rastelli2018computationally} reflects a genuine problem rather than being merely an artifact of the proofs. It also shows that we can have good convergence of many quantities of interest even when the likelihood does not converge pointwise. We suspect that a more refined study of the perturbed algorithms in a weaker norm would be able to show this comparison theoretically. Unfortunately, as is often the case in MCMC theory, at the moment we merely note that better theoretical bounds seem to go along with better empirical performance, even when the theoretical bounds are quantitatively too weak to give a meaningful comparison for realistic examples.

Finally, we discuss the choices of $b$ and $\kappa$ and how they relate to our bounds on running times and errors. For the algorithm in \cite{rastelli2018computationally}, one is free to choose any value of $b$, and (as discussed in the previous paragraph) the error bounds don't give effective guidance on the best choice of $b$. The running time of the algorithm increases as $b$ decreases, so we expect practitioners to choose the smallest value of $b$ that fits in their computational budget. For our algorithms, one is again free to choose any value of $b$, and again the error bounds don't give effective guidance, and so practitioners might make the same choices as they do with the algorithm in \cite{rastelli2018computationally}. However, our main error bound for Algorithm \ref{alg_MCMC_full} \textit{does} give effective bounds under some conditions. Very roughly, we expect to require $b \gtrsim |V|^{-c}$ for some $c> 0$ that depends on the underlying model. Remark \ref{RemDiscAssum3} gives a short explanation of how to find guarantees on the best value of $c$, but informally existing theory tells us that $c = \frac{1}{2}$ or $c=\frac{1}{3}$ are possible for some popular models and there exists some $c > 0$ for most reasonable models. The leading Taylor term is $O(M_g B_g^{\kappa+1}|V|^{2-c(\kappa+1)})$ when $b=|V|^{-c}$ and $\kappa$ is fixed, so our error bounds are then effective for $\kappa>\frac{2}{c}-1$. Thus, if one wants an MCMC algorithm whose stationary distribution converges to the true posterior in total variation as $|V|$ goes to infinity, one could use $b=|V|^{-\frac{1}{4}}$ and $\kappa=8$ to obtain an algorithm with amortized running time $O(|V|^{1.5}+|E|)$ and total variation error bound $O(M_g B_g^9 |V|^{-\frac{1}{4}})+\epsilon_{\mathrm{post}}$. More generally, for any fixed $C>0$ and admissible $c>0$, choosing a fixed $\kappa>\frac{C+2}{c}-1$ gives error $O(|V|^{-C})+\epsilon_{\mathrm{post}}$, up to the explicit factor $M_gB_g^{\kappa+1}$. Taking $c$ small makes the running time arbitrarily close to $O(|V|+|E|)$, with the caveat that the fixed constants can be large when $\kappa$ is large.

\subsection{Related Work}

There are several papers that focus on speeding up Bayesian inference for LPM or LPM-like models. We survey some of the main ideas and provide references to representative papers; see Section 4.3 of the recent survey \cite{kaur2023latent} for a more comprehensive overview. 

One natural approach is to replace MCMC by a variational approximation or by point estimates. Point estimates were already considered in \cite{hoff2002latent}, while \cite{Aliverti_2021,Liu_2023} propose a variational approximation for the LPM and \cite{salter2013variational} proposes variational approximations for the LPM-with-clustering model introduced in Handcock et al. \cite{handcock2007model}. 

Another approach is to carefully subsample observations. Generic subsampling algorithms such as stochastic gradient Langevin dynamics \cite{wellingteh11} can be easily used. However, a good understanding of what drives the embedding allows for more careful subsampling, as in the stratified case-control approximation of \cite{raftery2012fast}. It is also possible to adapt many of the advanced subsampling algorithms developed for simpler regression-like models to the case of LPMs. For example, \cite{spencer2022faster} adapts the ideas of split HMC \cite{shahbaba2014split} and the ``exact" subsampling algorithm of \cite{maclaurin2014firefly}.

The last approach, which most closely matches the approach taken in this paper, is to roughly chop up the latent space $\Omega$ into parts rather than subsampling edges.  Rastelli et al. \cite{rastelli2018computationally} take this approach, proposing a grid approximation of latent position for LPM and using a noisy Metropolis-within-Gibbs algorithm (see \textit{e.g.} \cite{alquier2016noisy}). Actual implementations of these algorithms require a careful choice of data structure, since there must be a computationally-efficient way to keep track of which vertices are in which ``part" of the latent space in order to achieve a real computational saving. The main difference between \cite{rastelli2018computationally} and the current paper comes from the approach for keeping track of this information: \cite{rastelli2018computationally} chops up the state space itself and thus must keep track of a map from vertices to parts of the state space, while we give an implicit partition of the node indices, allowing us to efficiently update a much finer approximation of the positions of all vertices.
 
There has been some success in finding conjugate priors to greatly speed up inference in some models related to LPMs \cite{ryan2017bayesian}, but to our knowledge this has not been possible for the LPM itself.

\subsection{Guide to Paper}

This paper is structured as follows: Section \ref{secNotation} introduces the main generic notation used throughout the paper. Section \ref{SecMainMomentCalcs} introduces our main likelihood approximation in terms of moments, while Section \ref{SecMainMomentDS} introduces the actual data structures and algorithms we use for efficiently keeping track of all necessary moments. Section \ref{SecAllBigAlgsInAlgNotation} collects our moment calculations and presents actual MCMC algorithms. Section \ref{SecAllTheory} analyzes our algorithms, providing bounds on the errors and computational complexities of the algorithms. Finally, Section \ref{sec-exp} gives empirical results and Section \ref{sec-conc} gives a short conclusion. 

\section{Notation} \label{secNotation}

We give the main generic notation used throughout the paper.

\subsection{Set and Multi-Index Notation}

For a positive integer $n$, let $[n]=\{1,2,\dots,n\}$. We write $\mathbb{N}_{0}=\{0,1,2,\ldots\}$ for the nonnegative integers.

We use multi-indices with either two coordinates, for single latent positions in $\mathbb{R}^{2}$, or four coordinates, for ordered pairs of latent positions in $(\mathbb{R}^{2})^{2}$. When $\alpha$ and $\beta$ are multi-indices of the same length, we write $\alpha\leq \beta$ if and only if $\alpha_i\leq \beta_i$ for every coordinate $i$ of the multi-index.

For the same $\alpha,\beta$, we also define

\be
|\alpha|=\sum_{i} \alpha_i, \quad \alpha!= \prod_{i} (\alpha_{i}!), \quad 
\binom{\alpha}{\beta}
=\prod_{i}\binom{\alpha_i}{\beta_{i}}.
\ee 

If $x$ has the same number of coordinates as $\alpha$, define

\be 
x^{\alpha} = \prod_{i} x(i)^{\alpha_{i}}.
\ee 

If $f$ is a function of variables indexed by the coordinates of $\alpha$, define the $\alpha$'th mixed partial derivative by
\be 
D^{\alpha} f = \left(\prod_{i} \frac{\partial^{\alpha_{i}}}{\partial x_{i}^{\alpha_{i}}} \right)f. 
\ee

\subsection{Graph and Likelihood Notation}

We introduce the main notation for graphs, embeddings and likelihoods that will be used until the full algorithms are presented in Section \ref{SecAllBigAlgsInAlgNotation}. The notation in this section is aimed at facilitating the main contributions of this paper, a system of likelihood approximations based on a carefully-chosen set of moments. As such, the notation purposefully elides some elements that are not relevant to moment calculations (such as the hyperparameters associated with the model), and allows for very generic choices even if they would lead to poor performance (for example, we allow for a completely generic partition $a$ of the graph's nodes). 

In the introduction, we introduced the simple LPM in Definition \ref{DefLPM}. There are many more complicated versions of this model in the literature (see \textit{e.g.} \cite{raftery2012fast}, which includes observed features for each node). Many of the ideas in this paper could be used for those models, but we restrict focus to only the small generalization:

\begin{defn} \label{DefLPM2}
Fix a latent space $\Omega$, distribution $\pi_{\Omega}$ on $\Omega$, parameterized family $\{p_{\theta}\}_{\theta \in \Theta}$ of symmetric link functions $p_{\theta} \, : \, \Omega^{2} \to [0,1]$, meaning $p_{\theta}(x,y)=p_{\theta}(y,x)$ for all $x,y\in\Omega$, prior $\pi_{\Theta}$ on $\Theta$, and size $n$. The \textit{latent position model} associated with the choices $(\Omega, \pi_{\Omega},\{p_{\theta}\},\pi_{\Theta},n)$ is a distribution on parameter $\theta$, graphs $G=(V,E)$ and latent positions $Z_{1},\ldots,Z_{n} \in \Omega$, given by sampling $\theta \sim \pi_{\Theta}$, setting $V = [n]$, sampling $Z_{1},\ldots,Z_{n} \stackrel{iid}{\sim} \pi_{\Omega}$, and finally sampling the unordered edge set $E_{und} \subseteq \{\{i,j\}:1\leq i<j\leq n\}$ conditional on $Z_{1},\ldots,Z_{n}$ and $\theta$ according to the distribution:
\begingroup
\be \label{EqDefLPMLikBasic2}
\mathcal{P}[E_{und} | Z_{1},\ldots,Z_{n},\theta] = \prod_{\substack{1\leq i<j\leq n\\ \{i,j\} \in E}} p_{\theta}(Z_{i},Z_{j}) \prod_{\substack{1\leq i<j\leq n\\ \{i,j\} \notin E}} (1- p_{\theta}(Z_{i},Z_{j})).
\ee 
\endgroup
The products in this definition are over unordered pairs of distinct vertices. By a small abuse of notation, we will often call this the latent position model associated with link function $p$ (hiding the dependence on the other choices).
\end{defn}

From this point onward, whenever log-likelihoods or log-link functions are used, we restrict attention to parameter values and latent positions for which $0<p_{\theta}(x,y)<1$ on the relevant domain. This is the standing condition needed for the quantities $g_1(x,y)=\log p_{\theta}(x,y)$ and $g_0(x,y)=\log(1-p_{\theta}(x,y))$ to be finite.

Throughout this paper, we consider an undirected graph $G=(V,E_{\mathrm{und}})$ with node set $V = [n]$. We will always denote by $\tau$ an embedding of $n$ points in $\mathbb{R}^{2}$, so that $\tau_{i} \in \mathbb{R}^{2}$ for each $i \in [n]$. We reserve $Z=(Z_1,\ldots,Z_n)$ for the latent positions in Definitions \ref{DefLPM} and \ref{DefLPM2}; $\tau$ denotes a generic embedding, posterior draw, or chain state. For likelihood and moment calculations, it is convenient to use ordered pairs of distinct vertices. Define
\be \label{EqDirectedPairSet}
\mathcal{P}_{n} &= \{(i,j) \in [n]^{2}: i \neq j\};\\
\qquad E &= \{(i,j) \in \mathcal{P}_{n}: \{i,j\} \in E_{\mathrm{und}}\};\\
E^{0}&=\mathcal{P}_{n}\setminus E.
\ee
Thus each unordered edge and each unordered nonedge appears twice in $E$ and $E^{0}$, respectively. Since the link functions are symmetric, this directed convention is equivalent to the unordered graph likelihood.

Our MCMC steps always update \textit{either} an embedding \textit{or} a parameter - never both. For that reason, we treat $\theta$ as fixed until we introduce the full algorithm in Section \ref{SecAllBigAlgsInAlgNotation}. When viewing $\theta$ as fixed, we introduce the shorthand $g_{1}(x,y) = \log (p_{\theta}(x,y))$ and $g_{0}(x,y) = \log (1- p_{\theta}(x,y))$. This lets us write the  log-likelihood of the model \eqref{EqDefLPMLikBasic2} as:

\be \label{EqLogLikSimple}
L(\tau) = \frac{1}{2}\sum_{(i,j) \in E} g_{1}(\tau_{i}, \tau_{j}) +\frac{1}{2} \sum_{(i,j) \in E^{0}} g_{0}(\tau_{i},\tau_{j}).
\ee 

When viewing the full algorithm with non-fixed $\theta$, we write $L_{\theta}(\tau)$ for the same function.

All of our algorithms require a choice of a partition $a \, : \, [n] \to [K]$ of the vertices of the graph and a Taylor order $\kappa$.  The multi-indices associated with pairs of points will be denoted by
\be \label{EqDefASet}
A=A_{\kappa}=\{\alpha=(\alpha_1,\alpha_2,\alpha_3,\alpha_4)\in\mathbb{N}_{0}^{4}: |\alpha|\leq \kappa\}
\ee
and the projected sets for single points will be denoted by
\be \label{EqAProjDef}
A_{proj}=A_{proj,\kappa} =\{\gamma=(\gamma_1,\gamma_2)\in\mathbb{N}_{0}^{2}: |\gamma|\leq \kappa\}.
\ee
We describe how $a$ and $\kappa$ may be chosen in Section \ref{SecPreproc}; until then, we treat them as generic.

\section{Main Moment Calculations}\label{SecMainMomentCalcs}

Our goal is to write down an estimator for the function $L$ in Equation \eqref{EqLogLikSimple} in terms of a relatively small collection of moments. We introduce the ``main" moments used, and write down a Taylor expansion in terms of these moments. In this section, we use the notation ``$f \approx g$" to indicate that $g$ is formally obtained by truncating a series expansion of $f$, without stating conditions under which $f$ and $g$ are actually pointwise close to each other. These calculations serve to motivate our algorithm, which is analyzed more carefully in Section \ref{SecAllTheory}.

\subsection{Main Moments}

We introduce the moments used in this approximation. It turns out to be computationally difficult to keep \textit{only} these moments updated while running the algorithm; in Section \ref{SecDStoKeep}, we will introduce a slightly larger family of moments that are easier to update together.

For indices $s,t \in [K]$ and multi-index $\alpha \in A$, define:

\be \label{Moment_once}
M_{s,t,\alpha}(\tau) = \sum_{(i,j) \in P_{s,t}} \tau_{i}(1)^{\alpha_{1}} \tau_{i}(2)^{\alpha_{2}}  \tau_{j}(1)^{\alpha_{3}}  \tau_{j}(2)^{\alpha_{4}}, 
\ee 
where $P_{s,t} = \{(i,j) \in E \, : \, a(i) = s, a(j) = t\}$. These are the mixed moments associated with observed edges. Similarly, define the moments associated with ``missing" edges:

\be 
M_{s,t,\alpha}^{(0)}(\tau) =  \sum_{(i,j) \in P_{s,t}^{0}} \tau_{i}(1)^{\alpha_{1}} \tau_{i}(2)^{\alpha_{2}}  \tau_{j}(1)^{\alpha_{3}}  \tau_{j}(2)^{\alpha_{4}},
\ee 
where $P_{s,t}^{0} = \{(i,j) \in E^{0} \, : \, a(i) = s, a(j) = t\}$. Finally, define the moments associated with the complete  graph without self-loops:

\be\label{Moment_total}
M_{s,t,\alpha}^{(c)}(\tau) &=  \sum_{(i,j) \in \mathcal{P}_{n} \, : \, a(i) = s, a(j) = t} \tau_{i}(1)^{\alpha_{1}} \tau_{i}(2)^{\alpha_{2}}  \tau_{j}(1)^{\alpha_{3}}  \tau_{j}(2)^{\alpha_{4}}.
\ee 

It is natural to write estimates for Equation \eqref{EqLogLikSimple} in terms of moments of the form $M_{s,t,\alpha}$ and $M_{s,t,\alpha}^{(0)}$. However, in practice, it is easier to keep  $M_{s,t,\alpha}$ and $M_{s,t,\alpha}^{(c)}$ updated. With the no-self-loop convention above, we have the exact identity

\be \label{EqTranslateMomentsCompleteMissing}
M_{s,t,\alpha}^{(0)}(\tau) = M_{s,t,\alpha}^{(c)}(\tau) - M_{s,t,\alpha}(\tau).
\ee 

The ``centers" of blocks are prominent in any discussion of Taylor expansions, so we give them special notation even though they do not need to be computed separately. For $s \in [K]$, define the vector $y_{s}$ by 

\be \label{EqDefCenters}
y_{s}(\tau) = \frac{1}{|\{i \, : \, a(i) = s\}|} \sum_{i\, :\, a(i)=s} \tau_i.
\ee

\subsection{Taylor expansions} \label{SecTaylor}

For this section, fix an embedding $\tau \in (\mathbb{R}^{2})^{n}$. For each $s \in [K]$, fix a point $y_s \in \mathbb{R}^2$; we will typically choose $y_{s}$ according to Equation \eqref{EqDefCenters}, though this is not necessary for the development in the current section.  

We recall the usual $\kappa$-th Taylor expansion of $g_1$ around $(y_s,y_t)$, written in compact multi-index notation:
\be
g_1(\tau_i,\tau_j)
\;\approx\;
\sum_{\substack{|\alpha|\leq \kappa}}
\frac{D^{\alpha}g_1(y_s,y_t)}{\alpha!}\,
\left[ (\tau_i,\tau_j)-(y_s,y_t) \right]^{\alpha}.
\ee
Since 
\be
\left[ (\tau_i,\tau_j)-(y_s,y_t) \right]^{\alpha}=\sum_{\substack{\beta\in\mathbb{N}_{0}^{4} \\ \beta \le \alpha}}
\binom{\alpha}{\beta}\,
\bigl(-(y_s,y_t)\bigr)^{\alpha - \beta}\,
(\tau_i,\tau_j)^\beta,
\ee
we can expand this as:

\be
\sum_{(i,j)\in P_{s,t}} g_1(\tau_i,\tau_j)&\approx\sum_{\substack{|\alpha|\leq \kappa}}
\frac{D^{\alpha}g_1(y_s,y_t)}{\alpha!}\,\sum_{\substack{\beta\in\mathbb{N}_{0}^{4} \\ \beta \le \alpha}}
\binom{\alpha}{\beta}\,
\bigl(-(y_s,y_t)\bigr)^{\alpha - \beta}\,
\sum_{(i,j)\in P_{s,t}}(\tau_i,\tau_j)^\beta\\
&=\sum_{\substack{|\alpha|\leq \kappa}}
\frac{D^{\alpha}g_1(y_s,y_t)}{\alpha!}\,\sum_{\substack{\beta\in\mathbb{N}_{0}^{4} \\ \beta \le \alpha}}
\binom{\alpha}{\beta}\,
\bigl(-(y_s,y_t)\bigr)^{\alpha - \beta}\,M_{s,t,\beta}(\tau).
\ee
Using the same calculation and then applying Equation \eqref{EqTranslateMomentsCompleteMissing},
\be
\sum_{(i,j)\in P_{s,t}^{0}}g_0(\tau_i,\tau_j)&\approx\sum_{\substack{|\alpha|\leq \kappa}}
\frac{D^{\alpha}g_0(y_s,y_t)}{\alpha!}\,\sum_{\substack{\beta\in\mathbb{N}_{0}^{4} \\ \beta \le \alpha}}
\binom{\alpha}{\beta}\,
\bigl(-(y_s,y_t)\bigr)^{\alpha - \beta}\,
\sum_{(i,j)\in P_{s,t}^{0}}(\tau_i,\tau_j)^\beta\\
&=\sum_{\substack{|\alpha|\leq \kappa}}
\frac{D^{\alpha}g_0(y_s,y_t)}{\alpha!}\,\sum_{\substack{\beta\in\mathbb{N}_{0}^{4} \\ \beta \le \alpha}}
\binom{\alpha}{\beta}\,
\bigl(-(y_s,y_t)\bigr)^{\alpha - \beta}\,M^{(0)}_{s,t,\beta}(\tau)\\
&=\sum_{\substack{|\alpha|\leq \kappa}}
\frac{D^{\alpha}g_0(y_s,y_t)}{\alpha!}\,\sum_{\substack{\beta\in\mathbb{N}_{0}^{4} \\ \beta \le \alpha}}
\binom{\alpha}{\beta}\,
\bigl(-(y_s,y_t)\bigr)^{\alpha - \beta}\,(M^{(c)}_{s,t,\beta}(\tau)-M_{s,t,\beta}(\tau)).
\ee
To simplify later notation, for functions $g$ defined on a neighborhood of the relevant center pairs, vectors $y_{1}, y_{2}\in\mathbb{R}^{2}$, integers $\kappa$, and collections of points $\{M_{\beta}\}_{\beta \in \mathbb{N}_{0}^{4} \, : \, |\beta| \leq \kappa}$, we define:
\be \label{EqDefTSimpler}
T(g,y_{1},y_{2},\kappa,M):&=\sum_{\substack{|\alpha|\leq \kappa}}
\frac{D^{\alpha}g(y_{1},y_{2})}{\alpha!}\,\sum_{\substack{\beta\in\mathbb{N}_{0}^{4} \\ \beta \le \alpha}}
\binom{\alpha}{\beta}\,
\bigl(-(y_{1},y_{2})\bigr)^{\alpha - \beta}\,M_{\beta}.
\ee
This suggests the following approximation to the log-likelihood in Equation \eqref{EqLogLikSimple}:
\be\label{approx_log_like}
\tilde{L}(\tau)&=\frac{1}{2} \sum_{s,t} \left[ \sum_{\substack{|\alpha|\leq \kappa}}
\frac{D^{\alpha}g_1(y_s,y_t)}{\alpha!}\,\sum_{\substack{\beta\in\mathbb{N}_{0}^{4} \\ \beta \le \alpha}}
\binom{\alpha}{\beta}\,
\bigl(-(y_s,y_t)\bigr)^{\alpha - \beta}\,M_{s,t,\beta}(\tau)\right ]\\
&+ \frac{1}{2} \sum_{s,t} \left[\sum_{\substack{|\alpha|\leq \kappa}}
\frac{D^{\alpha}g_0(y_s,y_t)}{\alpha!}\,\sum_{\substack{\beta\in\mathbb{N}_{0}^{4} \\ \beta \le \alpha}}
\binom{\alpha}{\beta}\,
\bigl(-(y_s,y_t)\bigr)^{\alpha - \beta}\,(M^{(c)}_{s,t,\beta}(\tau)-M_{s,t,\beta}(\tau)) \right ]. \\
&=\frac{1}{2} \sum_{s,t} T(g_1,y_{s}(\tau),y_{t}(\tau),\kappa,\{M_{s,t,\cdot}(\tau)\})
+ \frac{1}{2} \sum_{s,t} T(g_0,y_{s}(\tau),y_{t}(\tau),\kappa,\{M_{s,t,\cdot}^{(0)}(\tau)\}).
\ee

\subsection{Efficient Calculations of Changes to Log-Likelihood} \label{SecEffLLChangeCalc}

If we change an embedding $\tau$ to another embedding $\tau^{*}$ that differs in only a single coordinate, it is possible to compute the change in Equation \eqref{approx_log_like} without explicitly recomputing the full value of $\tilde{L}$ at both points. We give the details of this computation.

Fix an embedding $\tau \in (\mathbb{R}^{2})^{n}$, an index $i \in [n]$, and a displacement $\delta \in \mathbb{R}^{2}$. Denote by $\tau^{*}$ the embedding:
\be 
\tau_{i}^{*} &= \tau_{i} + \delta \\
\tau_{j}^{*} &= \tau_{j}, \qquad j \neq i.\\
\ee 
As shorthand, for any function $f$ with domain $(\mathbb{R}^{2})^{n}$, we define $\Delta f = f(\tau^{*}) - f(\tau),$ so that e.g.
\be 
\Delta M_{s,t,\beta} &=  M_{s,t,\beta}(\tau^{*}) -  M_{s,t,\beta}(\tau),  \\
\Delta \tilde{L} &= \tilde{L}(\tau^{*}) -  \tilde{L}(\tau),
\ee 
and so on. Let $r=a(i)$ and define the affected ordered block set
\be
\mathcal{I}_{r}=\{(s,t)\in [K]^{2}:s=r \text{ or } t=r\}.
\ee
For $(s,t)\notin \mathcal{I}_{r}$, neither the center pair $(y_s,y_t)$ nor the moments $M_{s,t,\cdot}$ and $M^{(c)}_{s,t,\cdot}$ change. Therefore the change in the approximate log-likelihood is computed by directly recomputing the affected block terms:
\be
\Delta  \tilde{L} &=\frac{1}{2} \sum_{(s,t)\in \mathcal{I}_{r}} \Bigl[ T(g_1,y_{s}(\tau^{*}),y_{t}(\tau^{*}),\kappa,\{M_{s,t,\cdot}(\tau^{*})\})
- T(g_1,y_{s}(\tau),y_{t}(\tau),\kappa,\{M_{s,t,\cdot}(\tau)\})\Bigr]\\
&\quad +\frac{1}{2} \sum_{(s,t)\in \mathcal{I}_{r}} \Bigl[ T(g_0,y_{s}(\tau^{*}),y_{t}(\tau^{*}),\kappa,\{M^{(c)}_{s,t,\cdot}(\tau^{*})-M_{s,t,\cdot}(\tau^{*})\})\\
&\hspace{3.35cm} - T(g_0,y_{s}(\tau),y_{t}(\tau),\kappa,\{M^{(c)}_{s,t,\cdot}(\tau)-M_{s,t,\cdot}(\tau)\})\Bigr].\\
 \label{delta_L}
\ee 
Since $|\mathcal{I}_{r}|=O(K)$ and the evaluation of each $T$ term takes time $O(|A|^{2})$, evaluation of Equation \eqref{delta_L} takes time $O(K\,|A|^{2})$.

\section{Update Rules and Complexity Estimates} \label{SecMainMomentDS}

It is not straightforward to efficiently update all values of $M_{s,t,\alpha}(\tau)$ as the embedding $\tau$ changes. To this end, we introduce a slightly larger data structure to keep track of while the MCMC algorithm runs, then describe its ``update rules" and their complexity.

\subsection{ List of Data Structures to Keep} \label{SecDStoKeep}

In addition to the moments $\{M_{s,t,\alpha}\}_{s,t \in [K], \alpha \in A}$ and $\{M_{s,t,\alpha}^{(c)}\}_{s,t \in [K], \alpha \in A}$, our main moment-based algorithm, Algorithm \ref{alg_MCMC}, will keep track of two other families of moments.

The first is:

\be \label{Q_once}
Q_{i,t,\alpha}(\tau) = \sum_{j \in p_{i,t} } \tau_{j}(1)^{\alpha_{1}} \tau_{j}(2)^{\alpha_{2}},
\ee 
where $i$ ranges over $[n]$, $t$ ranges over $[K]$, $\alpha$ ranges over $A_{proj}$, and $p_{i,t} = \{j \, : \, (i,j) \in E, a(j) = t\}$.

The second is:

\be \label{Q_total}
Q^{(c)}_{t,\alpha}(\tau)=\sum_{j:a(j)=t} \tau_{j}(1)^{\alpha_{1}} \tau_{j}(2)^{\alpha_{2}},
\ee
where $t$ ranges over $[K]$ and $\alpha$ ranges over $A_{proj}$.

While we do \textit{not} keep track of them in a separate data structure, we also introduce notation for a similar variable:
\be
Q_{i,t,\alpha}^{(0)}(\tau) = \sum_{j \in p_{i,t}^{(0)} } \tau_{j}(1)^{\alpha_{1}} \tau_{j}(2)^{\alpha_{2}},
\ee
where $i$ ranges over $[n]$, $t$ ranges over $[K]$, $\alpha$ ranges over $A_{proj}$, and $p_{i,t}^{(0)} = \{j \, : \, (i,j) \in E^{0}, a(j) = t\}$. We do not retain these variables because they can be computed in terms of the retained variables. With the no-self-pair convention,
\be \label{EqQ0FromQcQ}
Q_{i,t,\alpha}^{(0)}(\tau)=Q^{(c)}_{t,\alpha}(\tau)-Q_{i,t,\alpha}(\tau)-\mathbf{1}_{\{a(i)=t\}}\tau_{i}(1)^{\alpha_{1}} \tau_{i}(2)^{\alpha_{2}}.
\ee

\subsection{Update Rules and Complexity Estimates} \label{EqUpdateRulesAndComplexity}

Fix notation for $\tau$, $k$, $\delta$, $\tau^{*}$ and the operator $\Delta$ as in Section \ref{SecEffLLChangeCalc}, and write $r=a(k)$. By a small abuse of notation, define the function 
\be 
\tau_{k}(\alpha_1,\alpha_2) = \tau_{k}(1)^{\alpha_{1}} \tau_{k}(2)^{\alpha_{2}},
\ee 
so that we have the convenient shorthand:
\be\label{delta_M}
\Delta \tau_k(\delta,\alpha_1,\alpha_2)=(\tau_{k}(1) + \delta(1))^{\alpha_{1}} (\tau_{k}(2) +\delta(2))^{\alpha_{2}} - \tau_{k}(1)^{\alpha_{1}} \tau_{k}(2)^{\alpha_{2}}.
\ee

In this section, we give formulas for quickly calculating $\{\Delta M_{s,t,\alpha}\}$, $\{\Delta M_{s,t,\alpha}^{(c)}\}$, $\{ \Delta Q_{i,t,\alpha}\}$ and $\{ \Delta Q_{t,\alpha}^{(c)}\}$ in terms of $k,\delta, \{M_{s,t,\alpha}(\tau)\}$, $\{M_{s,t,\alpha}^{(c)}(\tau)\}$, $\{Q_{i,t,\alpha}(\tau)\}$ and $\{Q_{t,\alpha}^{(c)}(\tau)\}$. We also describe the complexity of these updates. The resulting rules are presented together as algorithms in Section \ref{SecMomentAlgs}.

We begin with update formulas for observed edges.

\textbf{Updating $Q_{i,t,\alpha}$:} for all $i \in [n]$, if $(i,k) \notin E$ or $a(k) \neq t$, then $Q_{i,t,\alpha}(\tau^{*}) = Q_{i,t,\alpha}(\tau)$ and so $\Delta Q_{i,t,\alpha} = 0$. In the remaining case that $(i,k) \in E$ and $a(k) = t$, we have:
\be \label{delta_Q_once}
\Delta Q_{i,t,\alpha} = \Delta \tau_k(\delta,\alpha_1,\alpha_2).
\ee 
The computational complexity of these updates is $\deg(k) \equiv |\{\forall \ i \,  : \, (i,k) \in E\}|$ for each tuple $\alpha \in A_{proj}$.

\textbf{Updating $M_{s,t,\alpha}$:} Before writing updates, we note 
\be 
M_{s,t,\alpha}(\tau) = \sum_{i \, : \, a(i) = s} \tau_{i}(1)^{\alpha_{1}}\tau_{i}(2)^{\alpha_{2}} Q_{i,t,(\alpha_{3},\alpha_{4})}(\tau).
\ee 
We use the symmetry
\be \label{EqMSymmRule}
M_{s,t,\alpha}(\tau) = M_{t,s,(\alpha_{3},\alpha_{4},\alpha_{1},\alpha_{2})}(\tau).
\ee 
For $s\neq t$, the two possible affected orientations are:
\begin{enumerate}
\item \textbf{Case 1:} $a(k)=s$. In this case,
\be \label{delta_case_1}
\Delta M_{s,t,\alpha} = \Delta \tau_k(\delta,\alpha_1,\alpha_2) Q_{k,t,(\alpha_{3},\alpha_{4})} (\tau).
\ee 
\item \textbf{Case 2:} $a(k)=t$. In this case,
\be \label{delta_case_2}
\Delta M_{s,t,\alpha} = Q_{k,s,(\alpha_{1},\alpha_{2})}(\tau)\Delta \tau_k(\delta,\alpha_3,\alpha_4).
\ee 
\end{enumerate}
For the within-block case $a(k)=s=t$, we have
\be\label{delta_case_3}
\Delta M_{s,s,\alpha}=\Delta \tau_k(\delta,\alpha_1,\alpha_2)Q_{k,s,(\alpha_3,\alpha_4)}(\tau)+Q_{k,s,(\alpha_1,\alpha_2)}(\tau)\Delta \tau_k(\delta,\alpha_3,\alpha_4).
\ee
Each of these updates has cost $O(1)$ for a fixed $\alpha$. Updating all affected observed-edge block moments for a proposed move of node $k$ costs $O(K|A|)$ after the $Q$ moments have been initialized.

\textbf{Updating $Q^{(c)}_{t,\alpha}$:} if $a(k) \neq t$, then $Q^{(c)}_{t,\alpha}(\tau^{*}) = Q^{(c)}_{t,\alpha}(\tau)$. In the remaining case $a(k) = t$, we have the update
\be \label{delta_Q_total}
\Delta Q^{(c)}_{t,\alpha} =\Delta \tau_k(\delta,\alpha_1,\alpha_2).
\ee
The computational complexity is $O(1)$ for each $\alpha \in A_{proj}$.

We now give rules for the missing-edge quantities. These are obtained from the complete-graph quantities and the observed-edge quantities.

\textbf{Updating $Q^{(0)}_{i,t,\alpha}$:} if $a(k) \neq t$, then $Q^{(0)}_{i,t,\alpha}(\tau^{*}) = Q^{(0)}_{i,t,\alpha}(\tau)$ and so $\Delta Q^{(0)}_{i,t,\alpha} = 0$. In the remaining case $a(k)=t$, Equation \eqref{EqQ0FromQcQ} gives
\be \label{delta_Q_miss}
\Delta Q^{(0)}_{i,t,\alpha} = \Delta Q^{(c)}_{t,\alpha}-\Delta Q_{i,t,\alpha}-\mathbf{1}_{\{i=k\}}\Delta \tau_k(\delta,\alpha_1,\alpha_2).
\ee

\par\textbf{Updating $M_{s,t,\alpha}^{(c)}$:}
Before updating $M_{s,t,\alpha}^{(c)}$, we note the useful formula
\be \label{EqMcFromQc}
M_{s,t,\alpha}^{(c)}(\tau)&=Q^{(c)}_{s,(\alpha_{1},\alpha_{2})} (\tau) Q^{(c)}_{t,(\alpha_{3},\alpha_{4})} (\tau) - \mathbf{1}_{\{s=t\}}Q^{(c)}_{s,(\alpha_1+\alpha_3,\alpha_2+\alpha_4)}(\tau).
\ee
For $s\neq t$, the two possible affected orientations are:
\begin{enumerate}
\item \textbf{Case 1:} $a(k)=s$. In this case,
\be \label{delta_case_1_total}
\Delta M_{s,t,\alpha}^{(c)} = \Delta \tau_k(\delta,\alpha_1,\alpha_2) Q^{(c)}_{t,(\alpha_{3},\alpha_{4})}(\tau).
\ee 
\item \textbf{Case 2:} $a(k)=t$. In this case, 
\be \label{delta_case_2_total}
\Delta M_{s,t,\alpha}^{(c)} =  Q^{(c)}_{s,(\alpha_{1},\alpha_{2})}(\tau)\Delta \tau_k(\delta,\alpha_3,\alpha_4).
\ee 
\end{enumerate}
For the within-block case $a(k)=s=t$, Equation \eqref{EqMcFromQc} gives
\be \label{delta_case_3_total}
\Delta M_{s,s,\alpha}^{(c)} &=Q_{s,(\alpha_1,\alpha_2)}^{(c)}(\tau)\Delta \tau_k(\delta,\alpha_3,\alpha_4)+Q_{s,(\alpha_3,\alpha_4)}^{(c)}(\tau)\Delta \tau_k(\delta,\alpha_1,\alpha_2) \\
&\quad +\Delta \tau_k(\delta,\alpha_1,\alpha_2)\,\Delta \tau_k(\delta,\alpha_3,\alpha_4)-\Delta \tau_k(\delta,\alpha_1+\alpha_3,\alpha_2+\alpha_4).
\ee 
Each of these complete-graph moment updates has cost $O(1)$ for a fixed $\alpha$. Updating all affected complete-graph block moments for a proposed move of node $i$ costs $O(|A|)$.

\subsection{Update Rules and Complexity Estimates - Faster Algorithm}

In the special case that $A=A_1=\{ \alpha \in \mathbb{N}_{0}^{4} \, : \, |\alpha| \leq 1\}$, we do not need to maintain the auxiliary $Q$ moments and there are much faster update rules. Define 
\[
E_{i,t} = |\{ j \,: \, a(j) = t, \, (i,j) \in E\}|,
\]
noting that this does not depend on the embedding and can be precomputed. Then
\be \label{EqFastMDef}
M_{s,t,(1,0,0,0)} &= \sum_{i \, : \, a(i) = s} \sum_{j \, : \, a(j) = t, (i,j) \in E} \tau_{i}(1) \\ 
&= \sum_{i \, : \, a(i) = s} E_{i,t} \tau_{i}(1). \\ 
\ee 
Thus, when index $k$ with $a(k) = s$ is updated, we can update with the rule:
\be \label{EqFastMDelta}
\Delta M_{s,t,(1,0,0,0)} = E_{k,t} \Delta \tau_{k}(\delta,1,0).
\ee 
The analogous formula with $\Delta \tau_k(\delta,0,1)$ updates $M_{s,t,(0,1,0,0)}$, and the remaining first-order moments follow from the symmetry \eqref{EqMSymmRule}. The cost of each update is $O(1)$ and we must do one update per value of $t$, for a total update cost of $O(K)$.

Similarly, for $s\neq t$ we have
\be \label{EqFastMDelta_Total}
\Delta M_{s,t,(1,0,0,0)}^{(c)} = |\{j \, : \, a(j) =t \}| \Delta \tau_{k}(\delta,1,0),
\ee
with the analogous formula for the second coordinate and the symmetric orientations. The within-block complete-graph first-order moments are updated by the specialization of Equation \eqref{delta_case_3_total}.

\section{Algorithms} \label{SecAllBigAlgsInAlgNotation}

We collect the main algorithms used in this paper.

\subsection{Preprocessing Algorithms} \label{SecPreproc}

Our main MCMC algorithms, Algorithms \ref{alg_MCMC_full} and \ref{alg_MCMC_full2}, require a function $a$ that partitions  the vertices in such a way that $\| \tau_{i} - \tau_{j}\|$ is ``small" whenever $a(i) = a(j)$ (though the converse need not hold). Finding a ``good" partition or clustering of the vertices of a graph is a well-studied problem that depends on the details of the generating process (see \textit{e.g.} \cite{berahmand2025comprehensivesurveyspectralclustering} for a broad survey of related methods, and \cite{von2008consistency} for the fundamental theory establishing rates of convergence for a collection of classical methods under various conditions). For this reason, improving on existing graph partition results is beyond the scope of this article. In Appendix \ref{SubsecNodePart}, we show how to construct a good partition from any sufficiently-good node embedding (see Lemma \ref{IneqBasicPartitionProperties}) and present a fairly simple partitioning strategy that has good properties for many popular choices of link function $p$ and prior $\pi_{\Omega}$ on latent positions.

\subsection{Moment Algorithms} \label{SecMomentAlgs}

We collect the main moment-based subalgorithms. Before stating them, we recall some shorthand used throughout this section.

We recall a useful shorthand introduced in Section \ref{SecEffLLChangeCalc}, which is used throughout this section. When the ``new" and ``old" embeddings $\tau, \tau^{*}$ are clear from the context, we use the shorthand $\Delta f = f(\tau^{*}) - f(\tau)$ for any function $f$ with domain $(\mathbb{R}^{2})^{n}$. In particular, we have e.g.
\be 
\Delta M_{s,t,\beta} &=  M_{s,t,\beta}(\tau^{*}) -  M_{s,t,\beta}(\tau)  \\
\Delta \tilde{L} &= \tilde{L}(\tau^{*}) -  \tilde{L}(\tau)
\ee 
so that
\be
 M_{s,t,\beta}(\tau^{*}) &=  M_{s,t,\beta}(\tau)+\Delta M_{s,t,\beta}\\
 \tilde{L}(\tau^{*})&= \tilde{L}(\tau)+\Delta \tilde{L}.
\ee

We use this $\Delta$ operator frequently in the following algorithms as a reminder that we typically calculate new moments by calculating these changes. We also use $A_{proj}$ as in Equation \eqref{EqAProjDef} throughout.

In the update algorithms below, quantities evaluated at $\tau^*$ are proposed values. They are computed temporarily and are committed to the stored data structures only if the Metropolis-Hastings proposal is accepted. When Equation \eqref{delta_L} is evaluated after moving node $k$, the proposed centers are computed as
\be
y_{a(k)}(\tau^*)&=y_{a(k)}(\tau)+\frac{\tau_k^*-\tau_k}{|\{j:a(j)=a(k)\}|},\\ 
y_s(\tau^*)&=y_s(\tau)\qquad \qquad(s\neq a(k)).
\label{update_y}\ee
In Algorithm \ref{alg_MCMC}, \textsc{UpdateM} is called before \textsc{UpdateQ}; this is intentional because the formulas for the proposed observed-edge moments use the old $Q$ moments.

\begin{algorithm}[H]
\caption{InitializeQM}
\label{alg_pre}
\begin{algorithmic}[1]
\REQUIRE{$\tau, E, a, A$}
\FOR{$s \in [K]$}
    \STATE{compute $n_s=|\{i:a(i)=s\}|$}
    \STATE{compute $y_s(\tau)$ by \eqref{EqDefCenters}}
\ENDFOR
\FOR{$t \in [K]$, $\alpha \in A_{proj}$}
    \STATE{set $Q^{(c)}_{t,\alpha}$ by \eqref{Q_total}}
\ENDFOR
\FOR{$i\in [n]$, $t \in [K]$, $\alpha \in A_{proj}$}
    \STATE{set $Q^{(1)}_{i,t,\alpha}$ by \eqref{Q_once}}
\ENDFOR
\FOR{$s,t \in [K]$, $\alpha \in A$}
    \STATE{calculate $M^{(1)}_{s,t,\alpha}$ by \eqref{Moment_once}}
    \STATE{calculate $M^{(c)}_{s,t,\alpha}$ by \eqref{EqMcFromQc}}
\ENDFOR
\RETURN{$y,M^{(1)},M^{(c)},Q^{(1)},Q^{(c)}$}
\end{algorithmic}
\end{algorithm}

\begin{algorithm}[H]
\caption{InitializeM}
\label{alg_pre_M}
\begin{algorithmic}[1]
\REQUIRE{$\tau,E,a,A_1,\{E_{i,t}\}_{i\in[n],t\in[K]}$}

\FOR{$s \in [K]$}
    \STATE{compute $n_s=|\{i:a(i)=s\}|$ and $S_{s,r}=\sum_{i:a(i)=s}\tau_i(r)$ for $r\in\{1,2\}$}
    \STATE{compute $y_s(\tau)$ by \eqref{EqDefCenters}}
\ENDFOR

\FOR{$s,t \in [K]$}
    \STATE{set $M^{(1)}_{s,t,(0,0,0,0)}=\sum_{i:a(i)=s} E_{i,t}$}
    \STATE{set $M^{(1)}_{s,t,(1,0,0,0)}=\sum_{i:a(i)=s} E_{i,t}\tau_i(1)$ and $M^{(1)}_{s,t,(0,1,0,0)}=\sum_{i:a(i)=s} E_{i,t}\tau_i(2)$}
    \STATE{set $M^{(1)}_{s,t,(0,0,1,0)}$ and $M^{(1)}_{s,t,(0,0,0,1)}$ by the symmetry \eqref{EqMSymmRule}}
    \STATE{set $M^{(c)}_{s,t,\alpha}$ for $\alpha\in A_1$ by the first-order specialization of \eqref{EqMcFromQc} using $n_s,n_t,S_{s,1},S_{s,2},S_{t,1},S_{t,2}$}
\ENDFOR
\RETURN{$y,M^{(1)},M^{(c)}$}
\end{algorithmic}
\end{algorithm}

\begin{algorithm}[H]
\caption{UpdateQ}\label{alg_update_Q}
    \begin{algorithmic}
       \REQUIRE{$\tau,k,\tau^*,E,a,A,Q^{(1)},Q^{(c)}$}
        \FOR{$j \in [n]\setminus \{k\}$ with $(j,k)\in E$, and $\alpha \in A_{proj}$}
            \STATE{set $Q_{j,a(k),\alpha}(\tau^*)$ by \eqref{delta_Q_once}}
        \ENDFOR
    \FOR{$\alpha \in A_{proj}$}
        \STATE{set $Q^{(c)}_{a(k),\alpha}(\tau^*)$ by \eqref{delta_Q_total}}
    \ENDFOR
        \RETURN{$Q^{(1)}(\tau^*),Q^{(c)}(\tau^*)$}
\end{algorithmic}
\end{algorithm}

\begin{algorithm}[H]
\caption{UpdateM}\label{alg_update_M}
    \begin{algorithmic}
       \REQUIRE{$\tau,\theta,k,\tau^*,E,a,A,y,M^{(1)},M^{(c)},Q^{(1)},Q^{(c)}$}
      
        \STATE{Compute $y_{a(k)}(\tau^*)$ by \eqref{update_y} }
        \FOR{$t\in [K]\setminus \{a(k)\}$, $\alpha \in A$}
            \STATE{Compute $M_{a(k),t,\alpha}(\tau^*)$ by \eqref{delta_case_1} and $M_{t,a(k),\alpha}(\tau^*)$ by \eqref{delta_case_2}}
            \STATE{Compute $M^{(c)}_{a(k),t,\alpha}(\tau^*)$ by \eqref{delta_case_1_total} and $M^{(c)}_{t,a(k),\alpha}(\tau^*)$ by \eqref{delta_case_2_total}}
        \ENDFOR
        \FOR{$\alpha \in A$}
            \STATE{Compute $M_{a(k),a(k),\alpha}(\tau^*)$ by \eqref{delta_case_3}}
            \STATE{Compute $M^{(c)}_{a(k),a(k),\alpha}(\tau^*)$ by \eqref{delta_case_3_total}}
        \ENDFOR
        \STATE{Compute $\Delta \tilde{L}(\tau^*)$ by \eqref{delta_L}}
        \RETURN{$y(\tau^*),M^{(1)}(\tau^*),M^{(c)}(\tau^*),\Delta \tilde{L}(\tau^*)$}
\end{algorithmic}
\end{algorithm}

\begin{algorithm}[H]
\caption{UpdateM2}\label{alg_update_M_Faster}
    \begin{algorithmic}
       \REQUIRE{$\tau,\theta,k,\tau^*,E,a,A_1,y,M^{(1)},M^{(c)},\{E_{i,t}\}_{i\in[n],t\in[K]}$ }
       \STATE{Compute $y_{a(k)}(\tau^*)$ by \eqref{update_y} };
        \FOR{$t\in [K]\setminus \{a(k)\}$, $\alpha \in A_{1}$ with $\| \alpha \|_{1} =1$}
            \STATE{Compute $M_{a(k),t,\alpha}(\tau^*)$ and $M_{t,a(k),\alpha}(\tau^*)$ by \eqref{EqFastMDelta} and symmetry}
            \STATE{Compute $M^{(c)}_{a(k),t,\alpha}(\tau^*)$ and $M^{(c)}_{t,a(k),\alpha}(\tau^*)$ by \eqref{EqFastMDelta_Total} and symmetry}
        \ENDFOR
        \FOR{$\alpha \in A_{1}$ with $\| \alpha \|_{1}=1$}
            \STATE{Compute $M_{a(k),a(k),\alpha}(\tau^*)$ by \eqref{EqFastMDelta} and symmetry}
            \STATE{Compute $M^{(c)}_{a(k),a(k),\alpha}(\tau^*)$ by \eqref{EqFastMDelta_Total} and symmetry}
        \ENDFOR
        \STATE{Compute $\Delta \tilde{L}(\tau^*)$ by \eqref{delta_L}}
        \RETURN{$y(\tau^*),M^{(1)}(\tau^*),M^{(c)}(\tau^*),\Delta \tilde{L}(\tau^*)$}
\end{algorithmic}
\end{algorithm}

\subsection{Full MCMC Algorithms}

We present the two main MCMC algorithms. Both algorithms rely on choices of base transition kernels and the initial partition $a$. Since the algorithms can be run for any choice of $a$, in this section we view $a$ as fixed; recall our discussion on algorithms for choosing $a$ in Section \ref{SecPreproc}. The base transition kernels $q_{emb}, q_{par}$ must satisfy the following conditions.

The embedding proposal $q_{emb}$ should be a collection of Markov transition kernels $\{q_{emb,i}\}_{i\in[n]}$ on $(\mathbb{R}^{2})^{n}$ with the following two properties:

\begin{enumerate}
    \item each $q_{emb,i}$ is reversible with respect to Lebesgue measure on the $i$th coordinate, and
    \item if $\tau^{*} \sim q_{emb,i}(\tau,\cdot)$, then $\tau_{j}^{*}=\tau_j$ for every $j\neq i$.
\end{enumerate}

Assume that $\pi_{\Omega}$ has a density with respect to Lebesgue measure on $\Omega$, and that $\pi_{\Theta}$ has a density with respect to some base measure on $\Theta$. We use the same symbols $\pi_{\Omega}$ and $\pi_{\Theta}$ for these densities in the Metropolis-Hastings ratios below. The base transition kernel $q_{par}$ should be any Markov transition kernel that is reversible with respect to the chosen base measure on $\Theta$.

We now state our algorithms for sampling new embeddings: our ``fast" algorithm (Algorithm \ref{alg_MCMC}) and our ``faster" algorithm (Algorithm  \ref{alg_MCMC2}). In the update routines below, every object carrying the argument $\tau^*$ is a tentative value computed for the proposal; the current stored moments and centers are overwritten only after the Metropolis-Hastings proposal is accepted.

\begin{algorithm}[H]
\caption{MCMC1}\label{alg_MCMC}
    \begin{algorithmic}
       \REQUIRE{$i,\pi_{\Omega},\tau,\theta,E,a,A,y,M^{(1)},M^{(c)},Q^{(1)},Q^{(c)},\tilde{L}, q_{emb}$}
        \STATE{Sample $\tau^{*} \sim q_{emb,i}(\tau,\cdot)$}
        \STATE{Compute $y(\tau^*),M^{(1)}(\tau^*),M^{(c)}(\tau^*),\Delta \tilde{L}$=UpdateM($\tau,\theta,i,\tau^*,E,a,A,y,M^{(1)},M^{(c)},Q^{(1)},Q^{(c)}$)}
        \STATE{Set $Q^{(1)}(\tau^*),Q^{(c)}(\tau^*)$=UpdateQ($\tau,i,\tau^*,E,a,A,Q^{(1)},Q^{(c)}$)}
        \STATE{Sample $U \sim \mathrm{Unif}(0,1)$}
        \IF{$U < \min\{1,\exp(\Delta \tilde{L})\,\pi_{\Omega}(\tau^{*}_{i})/\pi_{\Omega}(\tau_{i})\}$}
            \STATE{Return $\tau^{*},y(\tau^{*} ),M^{(1)}(\tau^*),M^{(c)}(\tau^*),Q^{(1)}(\tau^*),Q^{(c)}(\tau^*),\tilde{L} + \Delta \tilde{L}$}
        \ELSE
            \STATE{Return $\tau,y, M^{(1)},M^{(c)},Q^{(1)},Q^{(c)},\tilde{L}$}
        \ENDIF
\end{algorithmic}
\end{algorithm}

\begin{algorithm}[H]
\caption{MCMC2}\label{alg_MCMC2}
    \begin{algorithmic}
       \REQUIRE{$i,\pi_{\Omega},\tau,\theta,E,a,A,y,M^{(1)},M^{(c)},\{E_{i,t}\},\tilde{L}, q_{emb}$}
        \STATE{Sample $\tau^{*} \sim q_{emb,i}(\tau,\cdot)$}
        \STATE{Compute $y(\tau^*),M^{(1)}(\tau^*),M^{(c)}(\tau^*),\Delta \tilde{L}$=UpdateM2($\tau,\theta,i,\tau^*,E,a,A,y,M^{(1)},M^{(c)},\{E_{i,t}\}$)}
        \STATE{Sample $U \sim \mathrm{Unif}(0,1)$}
        \IF{$U < \min\{1,\exp(\Delta \tilde{L})\,\pi_{\Omega}(\tau^{*}_{i})/\pi_{\Omega}(\tau_{i})\}$}
            \STATE{Return $\tau^{*},y(\tau^{*}), M^{(1)}(\tau^*),M^{(c)}(\tau^*),\tilde{L} + \Delta \tilde{L}$}
        \ELSE
            \STATE{Return $\tau,y, M^{(1)},M^{(c)},\tilde{L}$}
        \ENDIF
\end{algorithmic}
\end{algorithm}

For future reference, we write down the usual Metropolis-within-Gibbs sampler for updating parameters:

\begin{algorithm}[H]
\caption{MCMCP}\label{alg_MCMC_param}
    \begin{algorithmic}
       \REQUIRE{$\pi_{\Theta},\theta,A,y,M^{(1)},M^{(c)},\tilde{L}, q_{par}$}
        \STATE{Sample $\theta^{*} \sim q_{par}(\theta,\cdot)$}
        \STATE{Compute $\tilde{L}^{*}$ by Equation \eqref{approx_log_like} with parameter $\theta^{*}$ and the current summaries $y,M^{(1)},M^{(c)}$}
        \STATE{Sample $U \sim \mathrm{Unif}(0,1)$}
        \IF{$U < \min\{1,\exp(\tilde{L}^{*} - \tilde{L})\,\pi_{\Theta}(\theta^{*})/\pi_{\Theta}(\theta)\}$}
            \STATE{Return $\theta^{*},\tilde{L}^{*}$}
        \ELSE
            \STATE{Return $\theta,\tilde{L}$}
        \ENDIF
\end{algorithmic}
\end{algorithm}

We combine these to get an associated pair of ``full" algorithms, Algorithms \ref{alg_MCMC_full} and \ref{alg_MCMC_full2}:

\begin{algorithm}[H]
\caption{Fast MCMC Algorithm}\label{alg_MCMC_full}
\begin{algorithmic}
\REQUIRE{$\pi_{\Omega},\pi_{\Theta},\tau^{(0)},\theta^{(0)},E,a,A,T, q_{emb}, q_{par}$}
\STATE{$y,M^{(1)},M^{(c)},Q^{(1)},Q^{(c)}$=InitializeQM($\tau^{(0)},E,a,A$)}
\STATE{Compute $\tilde{L}$ by Equation \eqref{approx_log_like} with parameter $\theta^{(0)}$ and the initialized summaries}
\FOR{$t=1,2,\ldots,T$}
\STATE{$\theta^{(t)}, \tilde{L} =$MCMCP$(\pi_{\Theta},\theta^{(t-1)},A,y,M^{(1)},M^{(c)},\tilde{L}, q_{par})$}
\STATE{$\tau^{(t)} = \tau^{(t-1)}$}
    \FOR{$i=1,2,\dots,n$}
        \STATE{$\tau^{(t)},y,M^{(1)},M^{(c)},Q^{(1)},Q^{(c)},\tilde{L}$=MCMC1($i,\pi_{\Omega},\tau^{(t)},\theta^{(t)},E,a,A,y,M^{(1)},M^{(c)},Q^{(1)},Q^{(c)},\tilde{L},q_{emb}$)}
    \ENDFOR
\ENDFOR
\RETURN{$(\tau^{(T)},\theta^{(T)})$}
\end{algorithmic}
\end{algorithm}

\begin{algorithm}[H]
\caption{Faster MCMC Algorithm}\label{alg_MCMC_full2}
\begin{algorithmic}
\REQUIRE{$\pi_{\Omega},\pi_{\Theta},\tau^{(0)},\theta^{(0)},E,a,A,T, q_{emb}, q_{par}$ with $A=A_1$}
\STATE{Compute $\{E_{i,t}\}_{i\in[n],t\in[K]}$}
\STATE{$y,M^{(1)},M^{(c)}$=InitializeM($\tau^{(0)},E,a,A,\{E_{i,t}\}$)}
\STATE{Compute $\tilde{L}$ by Equation \eqref{approx_log_like} with parameter $\theta^{(0)}$ and the initialized summaries}
\FOR{$t=1,2,\ldots,T$}
\STATE{$\theta^{(t)},\tilde{L} =$MCMCP$(\pi_{\Theta},\theta^{(t-1)},A,y,M^{(1)},M^{(c)},\tilde{L}, q_{par})$}
\STATE{$\tau^{(t)} = \tau^{(t-1)}$}
    \FOR{$i=1,2,\dots,n$}
        \STATE{$\tau^{(t)},y,M^{(1)},M^{(c)},\tilde{L}$=MCMC2($i,\pi_{\Omega},\tau^{(t)},\theta^{(t)},E,a,A,y,M^{(1)},M^{(c)},\{E_{i,t}\},\tilde{L},q_{emb}$)}
    \ENDFOR
\ENDFOR
\RETURN{$(\tau^{(T)},\theta^{(T)})$}
\end{algorithmic}
\end{algorithm}

\section{Theoretical Guarantees} \label{SecAllTheory}

We collect our main theoretical results: the stationary measure is ``close" to the true posterior (Section \ref{sec-Thms}) and the algorithms are ``fast" (Section \ref{sec-Thms2}).

\subsection{Theoretical guarantees: error estimates }\label{sec-Thms}

We present guarantees on the accuracy of Algorithms \ref{alg_MCMC_full} and \ref{alg_MCMC_full2} under certain assumptions.

\subsubsection{Assumptions About the Latent Position Model}

Following \cite{rastelli2018computationally}, we restrict our attention to LPMs with compact two-dimensional state spaces and compact parameter spaces: 

\begin{assumption}\label{assum1}
 $\Theta$ is a compact subset of $\mathbb{R}^{d_{\Theta}}$ for some finite $d_{\Theta}$, and there exists $0<\mathcal{H}<\infty$ such that $\Omega\subset[-\mathcal{H},\mathcal{H}]\times[-\mathcal{H},\mathcal{H}]$.
\end{assumption}

\begin{remark}
Our algorithms can be run for high-dimensional LPMs, but low-dimensional LPMs are vastly more popular than high-dimensional LPMs and the relevant asymptotic regimes look quite different in these two cases. For this reason, we restrict our analysis only to low-dimensional LPMs.
\end{remark}

Continuing to follow \cite{rastelli2018computationally}, we make the following somewhat strong assumptions on our link function:

\begin{assumption}\label{assum2}
    Fix a family of link functions $\{p_{\theta}\}_{\theta \in \Theta}$. We assume:
    \begin{enumerate}
    \item \textbf{Depends Only on Distances}: \par $p_{\theta}$ is of the form:
\be \label{EqLinkForm}
p_{\theta}(x,y) = f(\|x-y\|^{2};\theta).
\ee 
        \item \textbf{Nondegenerate Probabilities}: \par for all $d \geq 0$ and $\theta \in \Theta$, $$f(d;\theta)\in [p^L,p^U] \subset (0,1).$$

        \item \begingroup\textbf{Analytic log-link bound}: For each fixed $\theta \in \Theta$, define
        $$g_{1,\theta}(x,y)=\log p_{\theta}(x,y), \qquad g_{0,\theta}(x,y)=\log(1-p_{\theta}(x,y)).$$
        There is an open neighborhood $\mathcal{U}\subset (\mathbb{R}^{2})^{2}$ of $[-\mathcal{H},\mathcal{H}]^{2}\times[-\mathcal{H},\mathcal{H}]^{2}$ such that, for every $\theta\in\Theta$, the functions $g_{0,\theta}$ and $g_{1,\theta}$ extend to $C^{\infty}$ functions on $\mathcal{U}$. Moreover, there are finite constants $M_g,\rho_g>0$ such that, for every integer $m\geq 0$,
        \be \label{assum4}
        D_{m,g}:=\max_{\ell\in\{0,1\}}\sup_{\theta\in\Theta}\sup_{z\in\mathcal{U}}\sup_{|\alpha|=m}|D^{\alpha}g_{\ell,\theta}(z)|\leq M_g\,m!\,\rho_g^{-m}. 
        \ee
        \endgroup
    \end{enumerate}
\end{assumption}

\begin{remark} [Weakening These Assumptions] \label{RemWeakening}
Our assumption that the log-link function is analytic is useful to obtain cleaner theorem statements. This allowed us to summarize our theory in Section \ref {SecAsymptoticsComp} as saying that Algorithm \ref{alg_MCMC_full} has total variation error of $O(|V|^{2}(B_g b)^{\kappa +1})$, where the implied constant does not depend on $|V|, b$ or $\kappa$.

In the common situation that we fix $\kappa$ and consider what happens as $|V|, b$ change, we can replace Inequality \eqref{assum4} by the much weaker assumption that $g_{\ell,\theta}(z)$ has continuous $\alpha$'th partial derivatives for all multi-indices $\alpha$ of size $|\alpha| \leq \kappa$. This substitution does not change any of the proofs, but does mean that the implied constant in the error bound may depend on the order $\kappa$.

\end{remark}

\begin{remark} [Satisfying Assumptions and Comparison to Assumptions of \cite{rastelli2018computationally}]
Our new assumptions are weaker in one way (we no longer require that $f$ be monotone) and typically stronger in another (even with the modification in Remark \ref{RemWeakening}, we require $f$ have $\kappa$ continuous partial derivatives). 

Our assumptions do hold for reasonable link functions used in practice. In particular, Inequality \eqref{assum4} holds for the Gaussian link \eqref{Gaussian_link} when the parameter space is compact, $\sigma$ is bounded away from zero, and $p_\theta$ and $1-p_\theta$ are uniformly bounded away from zero.
\end{remark}

Finally, we assume that our partition is ``good" in the sense that nodes in the same part of the partition should have similar latent positions. We first define a ``good" partition with respect to a particular embedding:

\begin{defn} \label{DefBGood}
Fix an embedding $\tau \in (\mathbb{R}^{2})^{n}$ and a partition $a \, : \, [n] \to [K]$. We say that $a$ is \textit{b-good} with respect to $\tau$ if 
\be 
\max_{s \in [K]} (\max_{i, j \in [n] \, : \, a(i) = a(j) = s} \| \tau_{i} - \tau_{j} \|) \leq b.
\ee 
\end{defn}

In practice, we don't think of $G$ and $\phi$ as fixed. Instead, we think of the observed graph $G$ as random, and the partition $\phi$ as being computed according to some strategy (\textit{e.g.} the strategy in Appendix \ref{SubsecNodePart}) based on the observed graph. Thus, it makes sense to use the following extended definition of ``goodness" in place of the deterministic Definition \ref{DefBGood}:

\begingroup
\begin{defn} \label{DefBGood2}
Fix $n \in \mathbb{N}$, a family of link functions $\{p_{\theta}\}_{\theta \in \Theta}$, and priors $\pi_{\Theta}, \pi_{\Omega}$. Fix a graph $G$ and denote by $\pi_{G}$ the posterior distribution on $(\tau,\theta)$ given the observed graph $G$ under the model in Definition \ref{DefLPM2}. When $\pi_G$ is applied to events depending only on $\tau$, as below, it denotes the corresponding $\tau$-marginal. Let $\phi$ be a function taking as input a graph and returning as output a partition of the vertices of the graph.

For constants $0 < \epsilon_{\mathrm{post}}, b < \infty$, say that $\phi$ is $(\epsilon_{\mathrm{post}}, b)$-good with respect to the graph $G$ and the priors $\pi_{\Theta},\pi_{\Omega}$ if

\be 
\pi_{G}[\phi(G) \text{ is b-good with respect to } \tau] > 1-\epsilon_{\mathrm{post}}.
\ee 
\end{defn}

If a strategy $\phi$ is good with respect to a particular graph $G$, then simply truncating the posterior distributions to embeddings $\tau$ that are $b$-good does not substantially influence the posterior. A convenient-to-state assumption is:

\begin{assumption}\label{assum3}
The function $\phi$ is $(\epsilon_{\mathrm{post}},b)$-good with respect to the observed graph $G$.
\end{assumption}

Since the graph $G$ is random under the LPM, it is also useful to assume that $\phi$ is good with high probability over the sampled graph:

\begin{assumption}\label{assum3_alt}
Fix notation as in Definition \ref{DefBGood2}. We say that $\phi$ is $(\epsilon_{\mathrm{post}},\delta_{\mathrm{graph}}, b)$-good with respect to the associated LPM model if, when $G$ is sampled from the LPM model, $\phi$ is $(\epsilon_{\mathrm{post}},b)$-good with respect to $G$ with probability at least $1-\delta_{\mathrm{graph}}$.
\end{assumption}
\endgroup

\begin{remark} [On typical values of $b$] \label{RemDiscAssum3}
It is natural to ask: for what values of $b$ and $\epsilon_{\mathrm{post}}$ should we expect Assumption \ref{assum3} to hold, and what algorithms $\phi$ achieve these bounds? If we can find a good embedding of the vertices, we can find a good partition of the vertices (see Lemma \ref{IneqBasicPartitionProperties}), so it is sufficient to answer the question: what is the rate of convergence $b = b(n)$ for point estimates of graph embedding algorithms?

This turns out to be a difficult question to answer in full generality, and there is a large literature on the subject. The first standard approach in the literature is to construct a ``Laplacian" matrix from the adjacency matrix of the graph, then embed the vectors of the graph using the eigenvectors of this Laplacian with the highest few eigenvalues. One can then use a perturbation bound (such as the Davis-Kahan theorem - see \textit{e.g.} \cite{yu2014usefulvariantdaviskahantheorem} for an explanation) to obtain error bounds for the procedure. In simple settings, this approach is known to achieve the usual statistical rate of convergence of $b(n) = O \left(\frac{1}{\sqrt{n}} \right)$ (see \textit{e.g.} Example 1 of \cite{von2008consistency}). In many other settings, small modifications of this algorithm are known to achieve nearly-optimal rates of $b(n) = O \left(\frac{polylog(n)}{\sqrt{n}} \right)$ (see \textit{e.g.} \cite{natik2021consistencyspectralseriation}). There is a great deal of ongoing research in finding accurate embeddings of graphs and related objects - see \textit{e.g.} \cite{lu2023contextual,Gallagher02072024} for two very different families of modern work in the area.

For the purposes of this paper, we don't need such a strong bound in order to achieve good results - Theorem \ref{Thm1} and the random-graph corollary following it give strong estimates for sufficiently large $\kappa$ as long as $b(n) = O(n^{-c})$ for \textit{some} $c > 0$. Such bounds can be achieved by directly applying perturbation bounds such as Theorem 16 of \cite{von2008consistency}. 

\end{remark}

\begin{remark} [Our Assumptions and Unimodality]
If we insisted on being able to find a single partition $a$ such that the strong condition ``$\| \tau_{i} - \tau_{j} \|$ is small if and only if $a(i)= a(j)$" holds for most embeddings $\tau$ sampled from the posterior, we would be very close to insisting that the posterior is unimodal. While unimodal posteriors are fairly common in practice, this would be a fairly strong assumption. Fortunately, we have two ways around this difficulty:

\begin{enumerate}
    \item Even as-written, we only insist on the much weaker condition ``$\| \tau_{i} - \tau_{j} \|$ is small if $a(i)= a(j)$." This condition turns out to be much weaker than unimodality - indeed, it can be satisfied even for posteriors with a moderate number of very strongly-separated modes. 
    \item While we insist on the above condition holding for a \textit{single} pre-chosen partition $a$, this insistence is largely to keep the focus of our (already somewhat-lengthy) algorithms on the main ideas. If we allowed some re-initializations of $a$ during the run of our algorithm, it is possible to accommodate a broader class of posterior distributions.
\end{enumerate}

See Section \ref{RemMulti} for further discussion.
\end{remark}

\subsubsection{Main Error Bounds} \label{SecMainErrorBounds}

We introduce notation for embeddings that ``respect" a partition $a$, then explain and state our main theoretical error bounds. 

\begingroup
We note that there are two important sources of randomness our results. First, the sampled graph $G$ is random; once $G$ is observed, the partition $a=\phi(G)$, the posterior $\pi_G$, and the stationary measure of the approximate chain are deterministic functions of $G$. Second, after conditioning on $G$, samples drawn from either $\pi_G$ or the approximate stationary measure are random. In particular, both the stationary measure of our algorithm and the true posterior are random functions of the graph. Our main results say, informally, that these two measures are close \textit{with high probability,} but allow some nonzero chance of failure. This should not be alarming, as \textit{e.g.} it is possible for the observed graph itself to be very ill-behaved, or even empty, with some small probability.
\endgroup

\begingroup
\begin{defn} [$(a,b)$-Respecting Embeddings]
Fix $n,K \in \mathbb{N}$, $a \, : \, [n] \to [K]$, and constant $b > 0$. Define the set of $(a,b)$-\textit{respecting} embeddings:
\be 
S_{a,b} = \{ \tau \in \Omega^{n} \, : \, \max_{s \in [K]} \max_{i, j \in [n] \, : \, a(i) = a(j)=s} \| \tau_{i} - \tau_{j} \| \leq b\},
\ee 
and say that a distribution $\pi$ on $\Omega^{n}$ is supported on $(a,b)$-respecting embeddings if 
\be 
\pi(S_{a,b}) = 1.
\ee 
\end{defn}
\endgroup

\begingroup
Denote by $\pi_{G}$ the posterior distribution on $(\tau,\theta)$ induced by the likelihood $L$ in Equation \eqref{EqLogLikSimple}. For fixed $a$, $b$, and $\kappa$, define the truncated approximate posterior $\hat{\pi}_{G}$ by
\be\label{EqTruncatedApproxPosterior}
\hat{\pi}_{G}(B)=\frac{\displaystyle \int_{B\cap (S_{a,b}\times\Theta)} \exp(\tilde{L}_{\theta}(\tau))\,\pi_{\Omega}^{\otimes n}(d\tau)\,\pi_{\Theta}(d\theta)}{\displaystyle \int_{S_{a,b}\times\Theta} \exp(\tilde{L}_{\theta}(\tau))\,\pi_{\Omega}^{\otimes n}(d\tau)\,\pi_{\Theta}(d\theta)}
\ee
for measurable $B\subset \Omega^n\times\Theta$. We suppress the dependence of $\hat\pi_G$ on $a,b$, and $\kappa$.
\endgroup

We split our main results into three estimates: a deterministic bound on the error due to Taylor expansions, a bound for fixed good partition, and the final result.

\begingroup
\begin{prop} \label{DetTaylorBound}
Fix $n,K,\kappa \in \mathbb{N}$ and $a \, : \, [n] \to [K]$. Under Assumptions \ref{assum1} and \ref{assum2}, for all $\theta\in\Theta$ and $\tau\in S_{a,b}$,
\begingroup
\be \label{boundlike}
\begin{aligned}
|\tilde{L}_{\theta}(\tau)-L_{\theta}(\tau)|&\leq R_{n,\kappa}(b),\\
R_{n,\kappa}(b)&:=\frac{n(n-1)}{2}M_g(B_g b)^{\kappa+1},\qquad B_g:=\frac{2}{\rho_g}.
\end{aligned}
\ee
\endgroup
Consequently,
\be \label{boundlike_mult}
e^{-R_{n,\kappa}(b)}\leq \frac{\exp(\tilde{L}_{\theta}(\tau))}{\exp(L_{\theta}(\tau))}\leq e^{R_{n,\kappa}(b)}.
\ee
\end{prop}

\begin{theorem} \label{Thm1}
Fix $G$, $a$, $b$, and $\kappa$. Under Assumptions \ref{assum1} and \ref{assum2}, if
\be
\pi_G(S_{a,b}\times\Theta)\geq 1-\epsilon_{\mathrm{post}},
\ee
then
\be \label{tv}
\|\pi_{G}-\hat{\pi}_{G}\|_{TV}\leq e^{2R_{n,\kappa}(b)}-1 + \epsilon_{\mathrm{post}}, 
\ee
where $R_{n,\kappa}$ is defined in Equation \eqref{boundlike}.
\end{theorem}

\begin{corollary} \label{CorrMain}
If $a=\phi(G)$ and $\phi$ is $(\epsilon_{\mathrm{post}},\delta_{\mathrm{graph}},b)$-good with respect to the LPM model, then
\be \label{tv_alt}
\mathbb{P}\left[\|\pi_{G}-\hat{\pi}_{G}\|_{TV}\leq e^{2R_{n,\kappa}(b)}-1 + \epsilon_{\mathrm{post}}\right] \geq 1 - \delta_{\mathrm{graph}},
\ee
where the probability is with respect to the sampled graph $G$ and $R_{n,\kappa}$ is defined in Equation \eqref{boundlike}.
\end{corollary}
\endgroup

\begingroup
\begin{proof}
    The proofs are in Section \ref{proof1}.
\end{proof}

\begin{remark-non} [Interpretation]
The term $e^{2R_{n,\kappa}(b)}-1$ is the deterministic Taylor approximation error accumulated over all vertex pairs, while $\epsilon_{\mathrm{post}}$ is the posterior probability that the partition fails to be $b$-good. The random-graph corollary adds only the graph-level failure probability $\delta_{\mathrm{graph}}$.
\end{remark-non}

\begin{remark} [On the Truncated Target and Computations]
The notation $\hat{\pi}_{G}$ above is deliberately used for the truncated approximate posterior. In an implementation, the same Metropolis-Hastings ratio applies if $\tilde{L}$ is evaluated as the extended approximate log-likelihood that equals Equation \eqref{approx_log_like} on $S_{a,b}$ and equals $-\infty$ outside $S_{a,b}$. Proposed states outside $S_{a,b}$ then have zero acceptance probability through the factor $\exp(\Delta\tilde{L})$, so no separate rejection step is part of the algorithmic description.
\end{remark}
\endgroup

\subsection{Theoretical guarantees on running times }\label{sec-Thms2}

We have the following bound on the main inner loop of Algorithm \ref{alg_MCMC_full}. In the usual case that sampling new parameters from the chain $q_{par}$ has negligible cost, this is the amortized cost (or, equivalently, asymptotic running time per sample) of Algorithm \ref{alg_MCMC_full}.

\begin{theorem} \label{ThemAmRunning}
The cost of running the innermost for loop of Algorithm \ref{alg_MCMC_full} is $O(|E| \, |A|+ n K \, |A|^{2})$.

\end{theorem}
\begin{proof}

The proof amounts to adding up the complexity estimates in Sections \ref{SecEffLLChangeCalc} and \ref{EqUpdateRulesAndComplexity}. Note that a single inner loop of Algorithm \ref{alg_MCMC_full} requires one call to Algorithm \ref{alg_MCMC} for every $k \in [n]$. Algorithm \ref{alg_MCMC} then calls Algorithms \ref{alg_update_Q} and \ref{alg_update_M}.

Inspecting Section \ref{EqUpdateRulesAndComplexity}, the running time of Algorithm \ref{alg_update_Q} is dominated by calls to the update rule in Equation \eqref{delta_Q_once}. For each fixed $k \in [n]$ these have a total cost of $O(|A| deg(k))$, and thus summing across $k \in [n]$ have a total cost of $O(|A| \, |E|)$.

Inspecting Section \ref{EqUpdateRulesAndComplexity} again, Algorithm \ref{alg_update_M} has three steps that could potentially dominate the running time: calls to Equation \eqref{delta_case_1}, calls to Equation \eqref{delta_case_3}, and calls to Equation \eqref{delta_L}. The remaining update rules are all update rules for moments associated with the ``complete" graph, which are much faster than the update rules for moments associated only with ``observed" edges.

Taking these in turn, Equation \eqref{delta_case_1} is called once per block $t$ and per $\alpha \in A$. Thus it has a running time of $O(K \, |A|)$ for each $k \in [n]$, and summing across $k \in [n]$ a total cost of $O(n K \, |A|)$. Equation \eqref{delta_case_3} is called once per $\alpha \in A$, so it has a running time of $O(|A|)$ for each $k \in [n]$, and summing across $k \in [n]$ a total cost of $O(n \, |A|)$. Finally, as computed at the end of Section \ref{SecEffLLChangeCalc}, evaluation of Equation \eqref{delta_L} has a running time of $O(K \, |A|^{2})$ for each $k \in [n]$, and summing across $k \in [n]$ a total cost of $O(n K \, |A|^{2})$. 

Combining all of the terms in the previous two paragraphs, the total running time is $O(|E| \, |A|+n K \, |A|^{2})$ as desired. 

\end{proof}

Next, we give analogous results for the likelihood used in Algorithm \ref{alg_MCMC2}:

\begin{theorem} \label{ThemAmRunning2}
If $A=A_1$, the cost of running the innermost for loop of Algorithm \ref{alg_MCMC_full2} is $O(n K)$.
\end{theorem}
\begin{proof}
The proof again amounts to adding up the complexity estimates in Sections \ref{SecEffLLChangeCalc} and \ref{EqUpdateRulesAndComplexity}. Note that a single inner loop of Algorithm \ref{alg_MCMC_full2} requires one call to Algorithm \ref{alg_MCMC2} for every $k \in [n]$. Algorithm \ref{alg_MCMC2} then calls Algorithm \ref{alg_update_M_Faster}.

Inspecting Algorithm \ref{alg_update_M_Faster}, the running time is dominated by calls to Equations \eqref{EqFastMDelta} and \eqref{delta_L}.

Taking these in turn, Equation \eqref{EqFastMDelta} is called once per block $t$ for every $k \in [n]$, and thus has a cost of $O(K)$ for each fixed $k \in [n]$ and a total cost of $O(nK)$. By the same analysis as in the proof of Theorem \ref{ThemAmRunning}, calls to Equation \eqref{delta_L} have a total cost of $O(n K )$. Combining these terms, the total cost is $O( n K )$ as desired. 

\end{proof}

\section{Empirical Results }\label{sec-exp}

In this section, we present three simulation studies to analyze the error introduced by our approximation algorithm. Since this paper is concerned with numerical approximation of a popular model (rather than introducing a new model), our empirical work focuses entirely on evaluating the strength of the approximation rather than the strength of the model.

We consider a latent position model (LPM) with global parameters $\beta_0$, $\beta_1$, and $\sigma$, where the edge probability is given by the Gaussian link
\begin{equation}\label{Gaussian_link_empirical}
p(\tau_i,\tau_j; \beta_0,\beta_1,\sigma)
= \beta_0 + \beta_1 \exp\!\left(-\frac{\|\tau_i-\tau_j\|^2}{2\sigma^2}\right).
\end{equation}
The latent positions are assumed to be IID and follow a truncated Gaussian distribution on $[0,1]^2$. The proposal distributions used in the sampler are also truncated Gaussian random walks, restricted to regions where the Taylor approximation error remains below a prescribed threshold.

\subsection{Rainbow plots} \label{SubsecRainbow}

In the first study, we compare the \textit{global} structure of the posterior distribution of the LPM to the stationary measures of the Markov chains defined in Algorithms~\ref{alg_MCMC_full} and \ref{alg_MCMC_full2}. Our goal is to check that, for the purposes of visualizing a full dataset, even this ``faster" algorithm does well.

We generate synthetic LPM networks with parameters $\beta_0=0.1$, $\beta_1=0.7$, and $\sigma=0.6$. We consider networks with $N \in \{2000,30000\}$ nodes. For each network size, we run the exact Metropolis-within-Gibbs (MwG) sampler for $10{,}000$ iterations, discard the first $5{,}000$ iterations as burn-in, and retain every $20$th sample thereafter. For Algorithm~\ref{alg_MCMC_full}, we use Taylor expansion order {$\kappa=4$}. 

To visually compare the recovered latent geometry, we align the posterior mean embedding obtained from each approximate method to the ground-truth posterior mean using Procrustes transformation. The nodes are then colored according to their true latent positions, so that distortions in the recovered geometry can be visually assessed.

\begin{figure}[htbp]
    \centering
    \begin{subfigure}{0.48\linewidth}
        \centering
        \includegraphics[width=\linewidth,height=0.32\textheight,keepaspectratio]{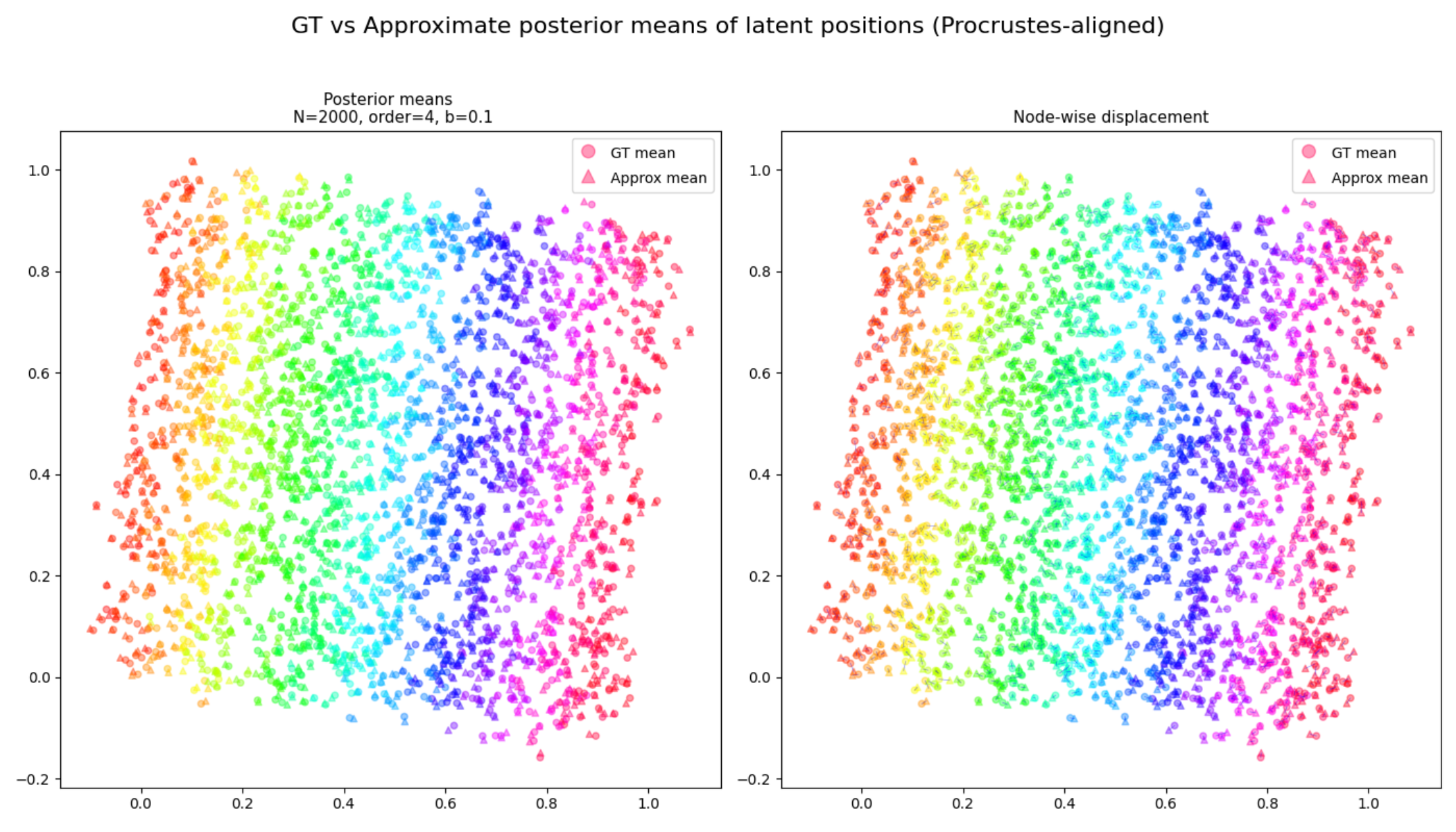}
        \caption{ Algorithm {\ref{alg_MCMC_full}} ({$\kappa=4, b=0.1$})}
        \label{fig:rainbow_order4}
    \end{subfigure}
    \hfill
    \begin{subfigure}{0.48\linewidth}
        \centering
        \includegraphics[width=\linewidth,height=0.32\textheight,keepaspectratio]{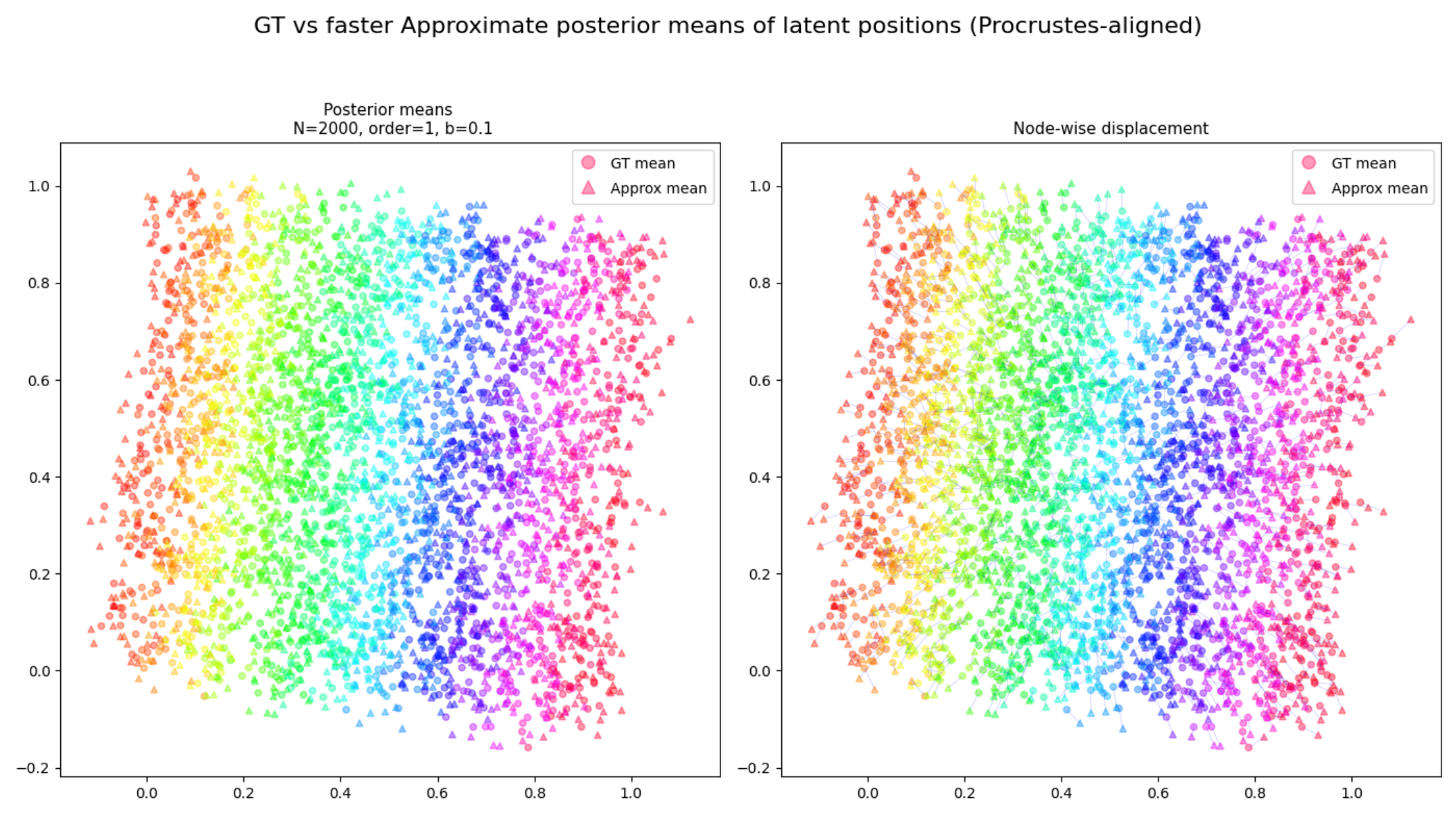}
        \caption{Algorithm {\ref{alg_MCMC_full2}} {$(\kappa=1, b=0.1$})}
        \label{fig:rainbow_order1}
    \end{subfigure}
    \caption{Rainbow plots comparing posterior mean latent positions obtained by the approximate samplers against the ground-truth MwG sampler. $N=2000$}
    \label{fig:rainbow_plots}
\end{figure}
\begin{figure}[htbp]
    \centering
    \begin{subfigure}{0.48\linewidth}
        \centering
        \includegraphics[width=\linewidth,height=0.32\textheight,keepaspectratio]{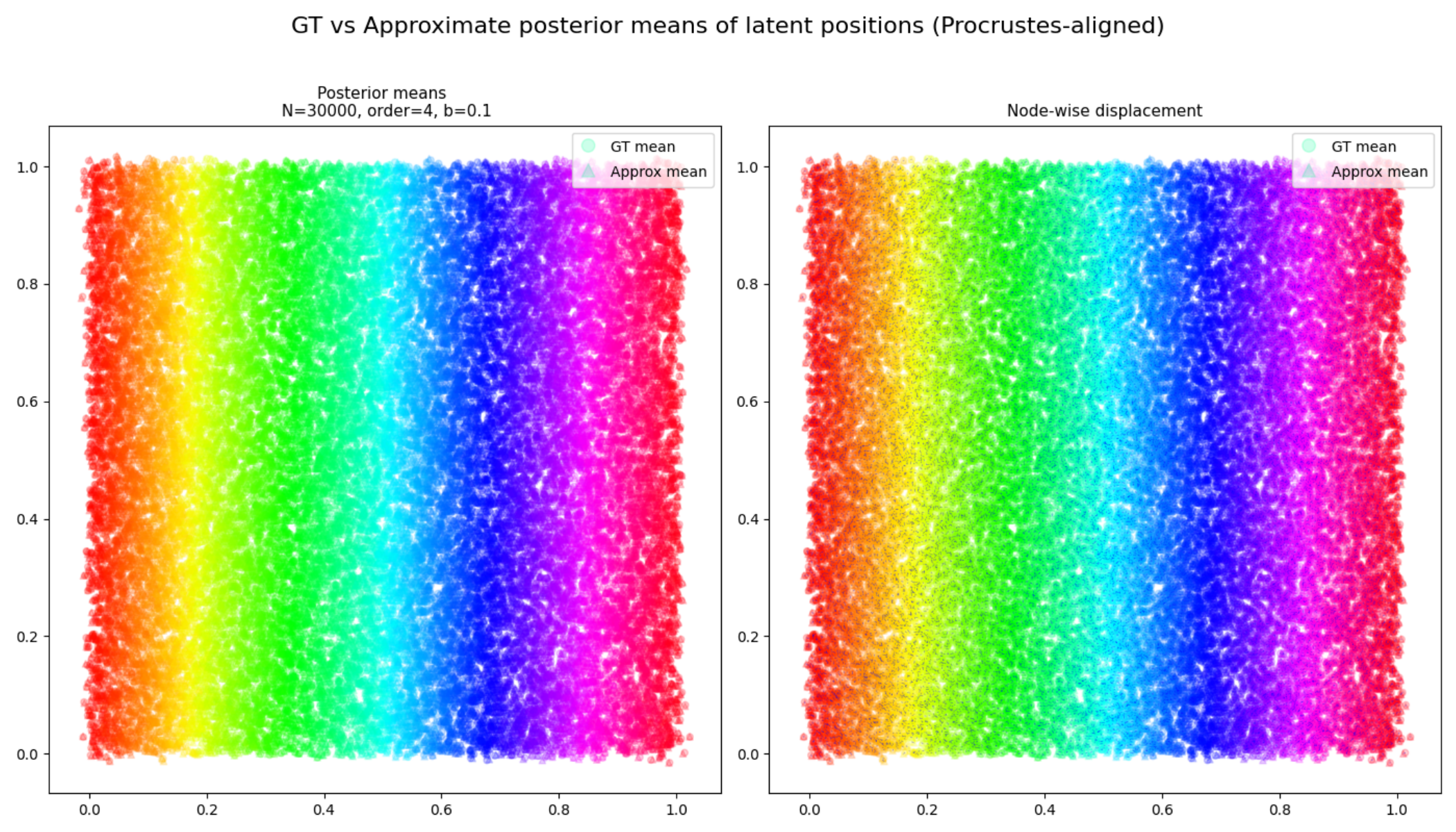}
        \caption{Algorithm {\ref{alg_MCMC_full}} ({$\kappa=4, b=0.1$})}
        \label{fig:rainbow_order4_2}
    \end{subfigure}
    \hfill
    \begin{subfigure}{0.48\linewidth}
        \centering
        \includegraphics[width=\linewidth,height=0.32\textheight,keepaspectratio]{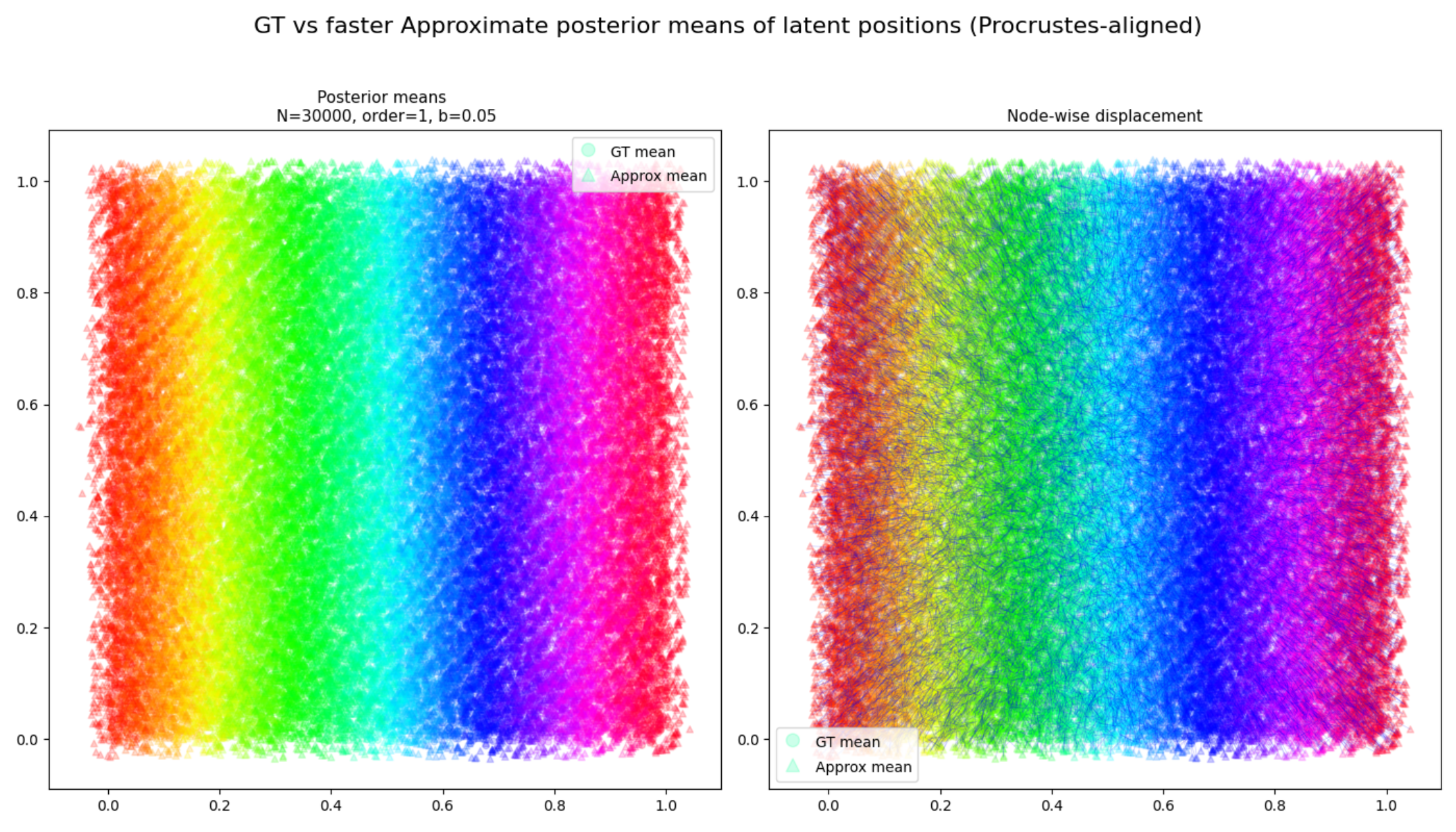}
        \caption{Algorithm {\ref{alg_MCMC_full2}} ({$\kappa=1, b=0.05$})}
        \label{fig:rainbow_order1_2}
    \end{subfigure}
    \caption{Rainbow plots comparing posterior mean latent positions obtained by the approximate samplers against the ground-truth MwG sampler. $N=30000$}
    \label{fig:rainbow_plots_2}
\end{figure}
Figures~\ref{fig:rainbow_plots} and~\ref{fig:rainbow_plots_2} compare the posterior mean latent positions obtained from the approximate samplers with those obtained from the ground-truth MwG sampler, after Procrustes alignment. The rainbow coloring provides a qualitative diagnostic of whether the global latent geometry is preserved.

For both N=2000 and N=30000, the approximate posterior means recover the main large-scale structure of the latent space. In particular, Algorithm~\ref{alg_MCMC_full} with $\kappa=4$ produces posterior mean embeddings that are visually close to the ground-truth MwG embedding. The faster first-order method, Algorithm~\ref{alg_MCMC_full2}, also captures the broad geometry, but exhibits more visible distortion, especially in the larger graph.

\begin{figure}[htbp]
    \centering
    \includegraphics[width=0.9\linewidth,height=1.5\textheight,keepaspectratio]{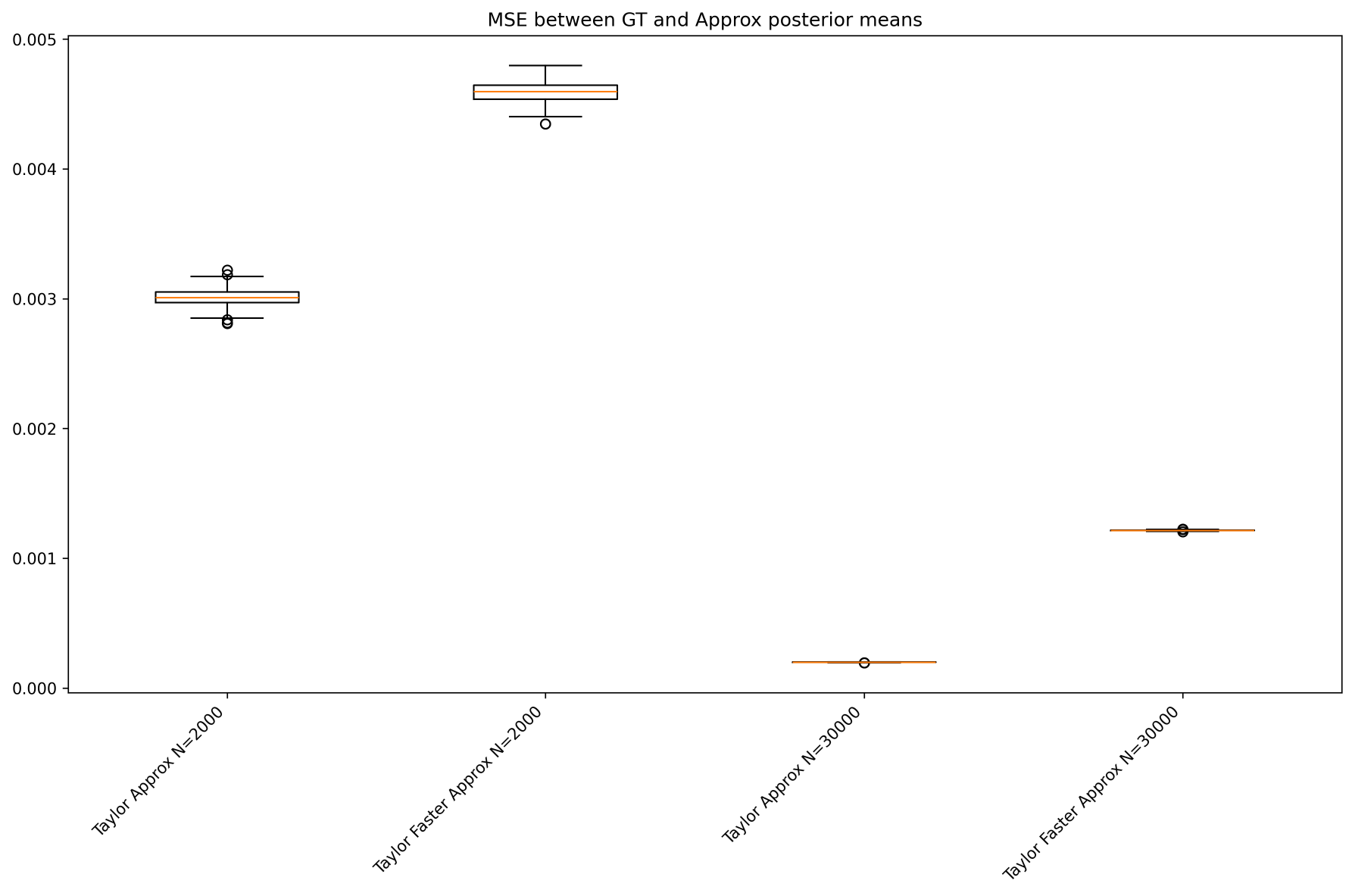}
    \caption{Mean squared error compared with the ground-truth posterior means.}
    \label{fig:MSE}
\end{figure}
The above visual conclusion is supported by the mean-squared-error comparison in Figure~\ref{fig:MSE}. The MSEs remain small across all settings, indicating that both approximate methods provide reasonable posterior mean estimates. However, the error for the faster first-order approximation is noticeably larger than that of the fourth-order method. This suggests that, although the first-order method is computationally attractive, its approximation error is no longer negligible relative to the posterior uncertainty in some settings.

\subsection{Single-node posterior contour}

In the second study, we fix the latent positions of all but one node and investigate the conditional distribution of the one unfrozen node. It is common to look at the per-point uncertainty in dimension-reduction algorithm, so this conditional distribution is important in its own right. It also allows us to more easily visualize errors in the likelihood, since high-dimensional functions are generally quite hard to visualize. \par
\begin{figure}[hbtp]
    \centering
    \captionsetup[subfigure]{font=small, labelfont=bf}

    \begin{subfigure}[b]{0.31\textwidth}
        \centering
        \includegraphics[width=\linewidth]{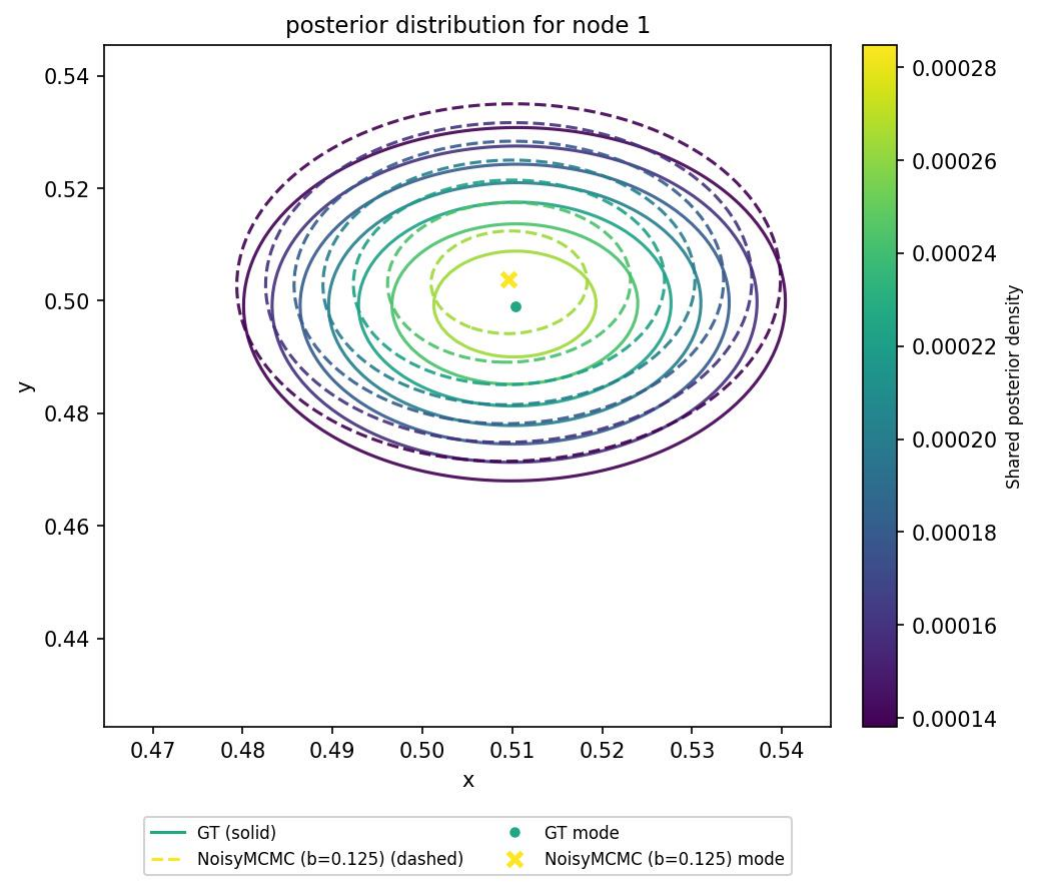}
        \caption{$b=0.125$, Noisy MCMC}
        \label{fig:single_node_b0125_friel}
    \end{subfigure}
    \hfill
    \begin{subfigure}[b]{0.31\textwidth}
        \centering
        \includegraphics[width=\linewidth]{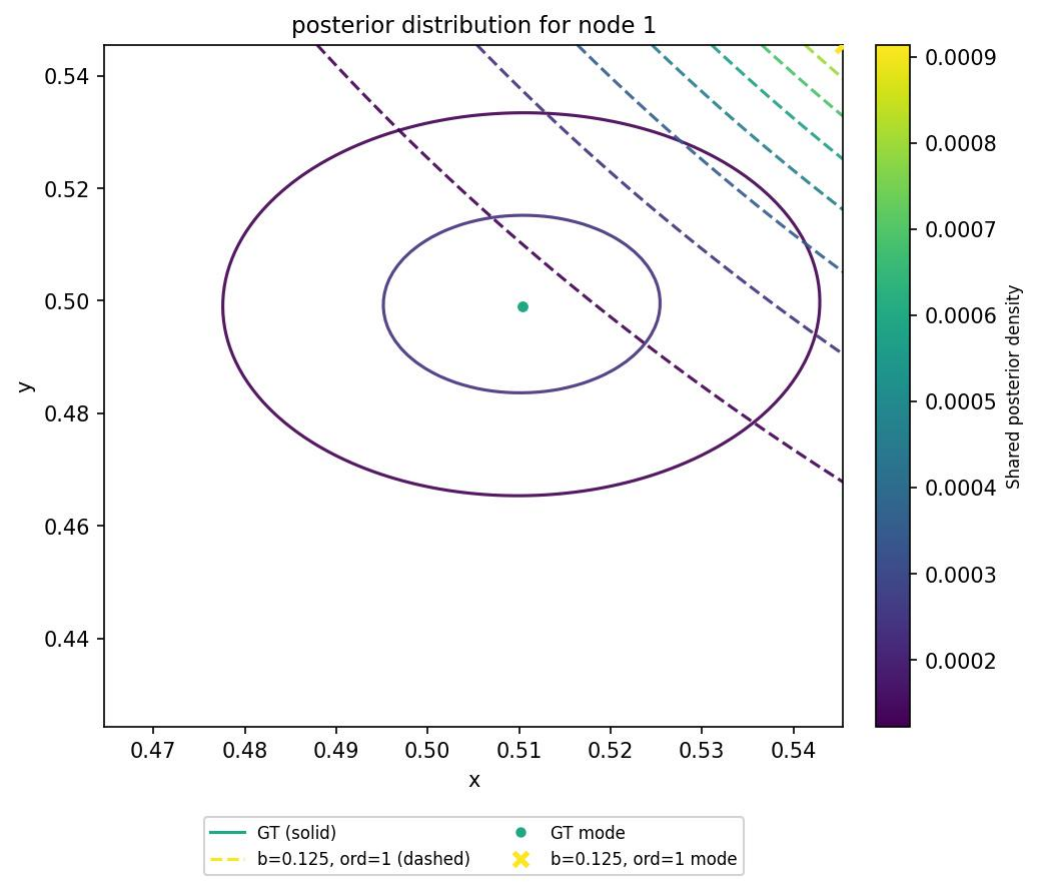}
        \caption{$b=0.125$, Algorithm {\ref{alg_MCMC_full2}}, {$\kappa=1$}}
        \label{fig:single_node_b0125_order1}
    \end{subfigure}
    \hfill
    \begin{subfigure}[b]{0.31\textwidth}
        \centering
        \includegraphics[width=\linewidth]{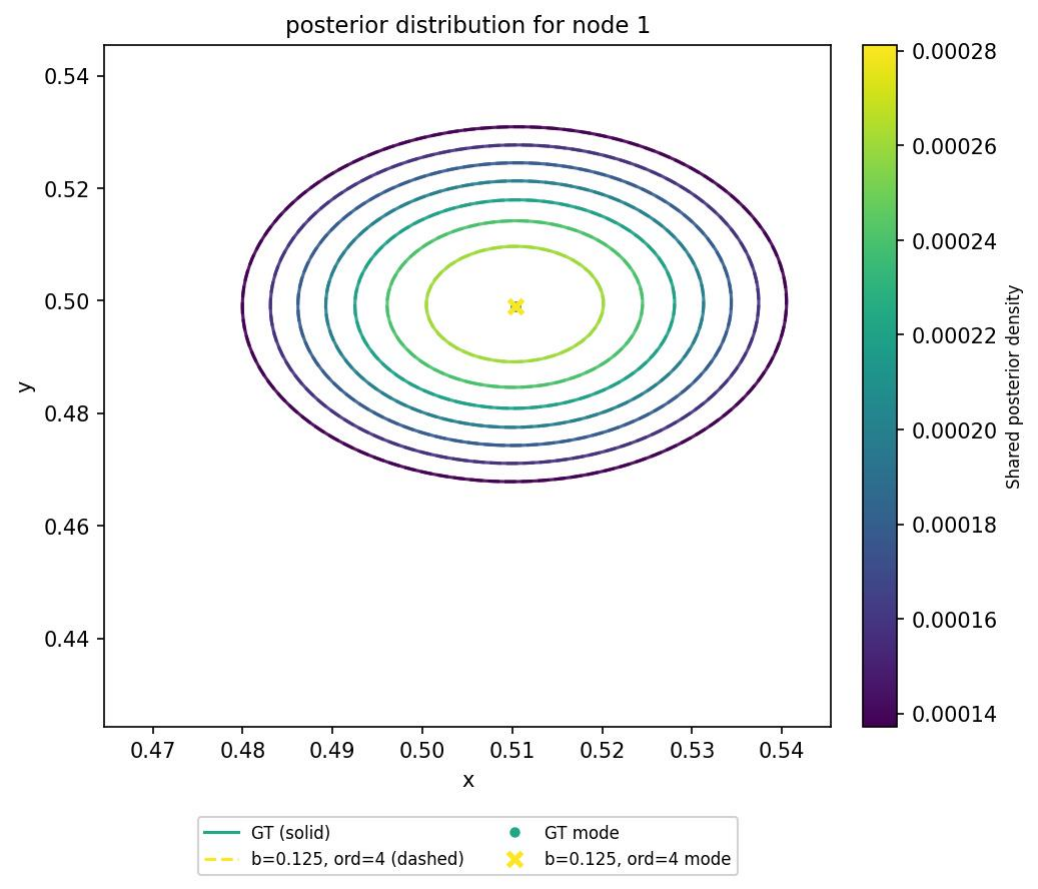}
        \caption{$b=0.125$, Algorithm {\ref{alg_MCMC_full}}, {$\kappa=4$}}
        \label{fig:single_node_b0125_order4}
    \end{subfigure}

    \vspace{0.6em}

    \begin{subfigure}[b]{0.31\textwidth}
        \centering
        \includegraphics[width=\linewidth]{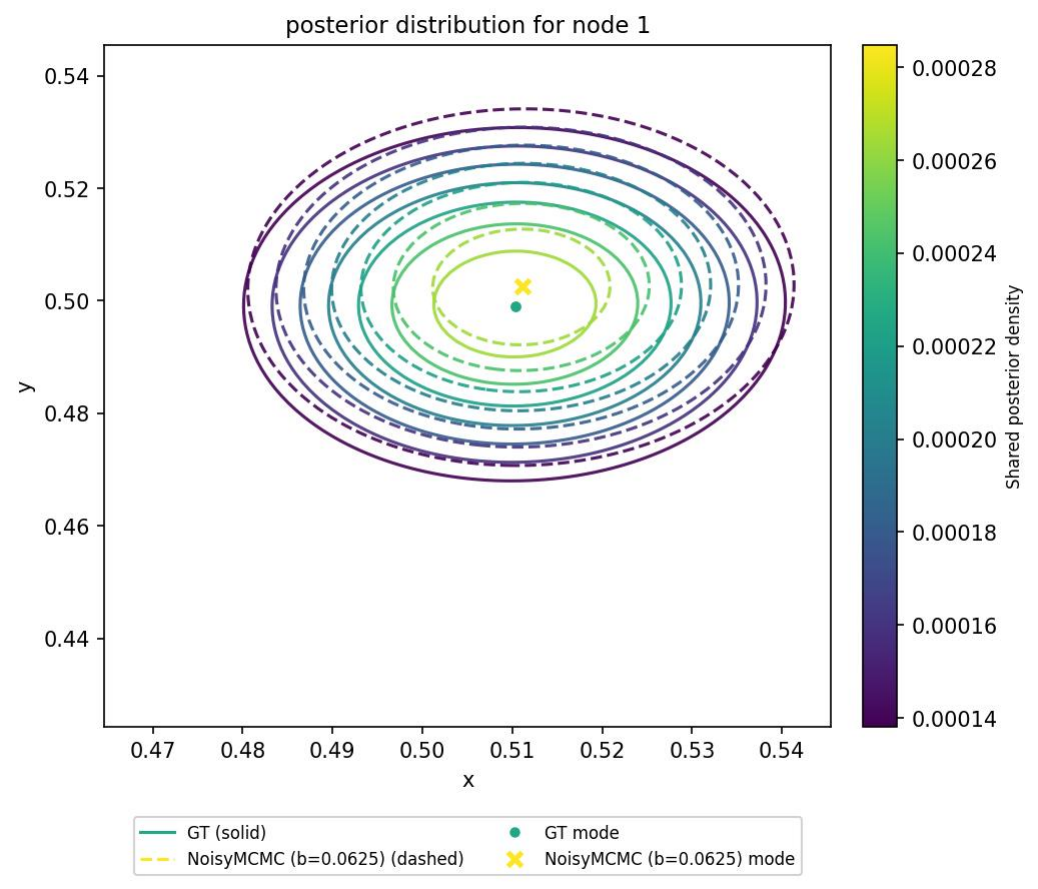}
        \caption{$b=0.0625$, Noisy MCMC}
        \label{fig:single_node_b00625_friel}
    \end{subfigure}
    \hfill
    \begin{subfigure}[b]{0.31\textwidth}
        \centering
        \includegraphics[width=\linewidth]{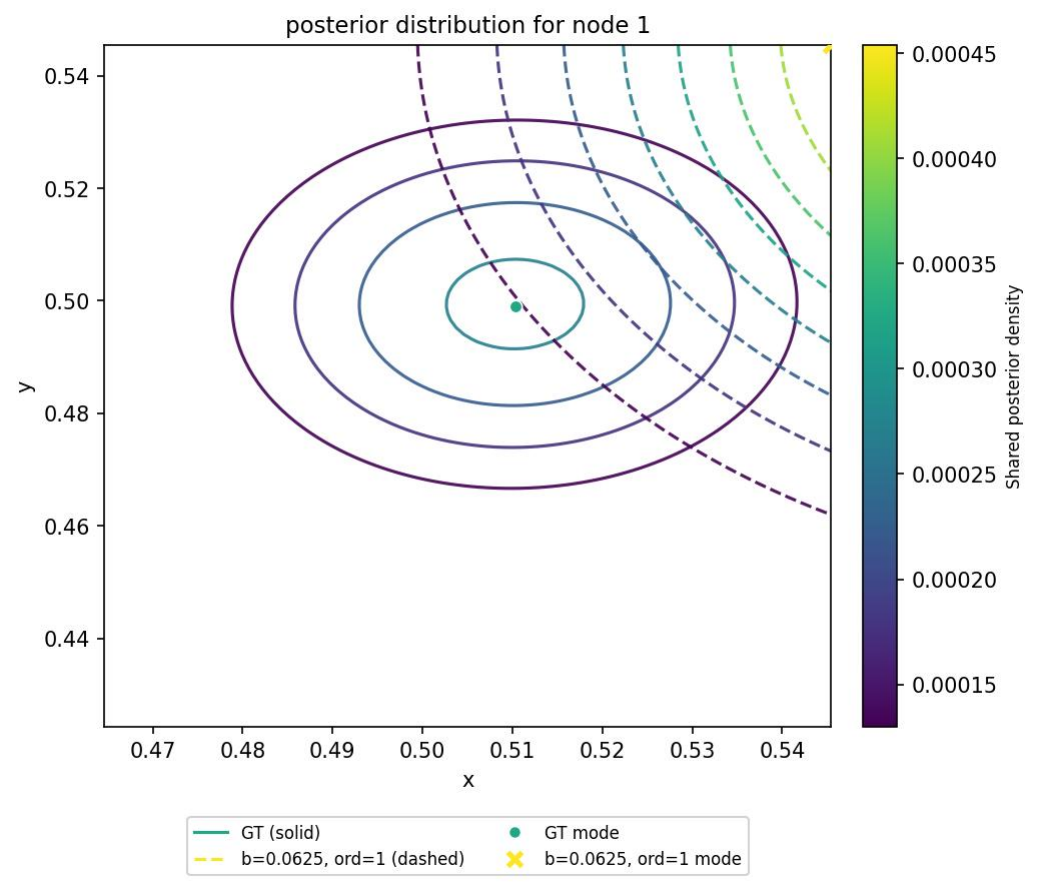}
        \caption{$b=0.0625$, Algorithm {\ref{alg_MCMC_full2}}, {$\kappa=1$}}
        \label{fig:single_node_b00625_order1}
    \end{subfigure}
    \hfill
    \begin{subfigure}[b]{0.31\textwidth}
        \centering
        \includegraphics[width=\linewidth]{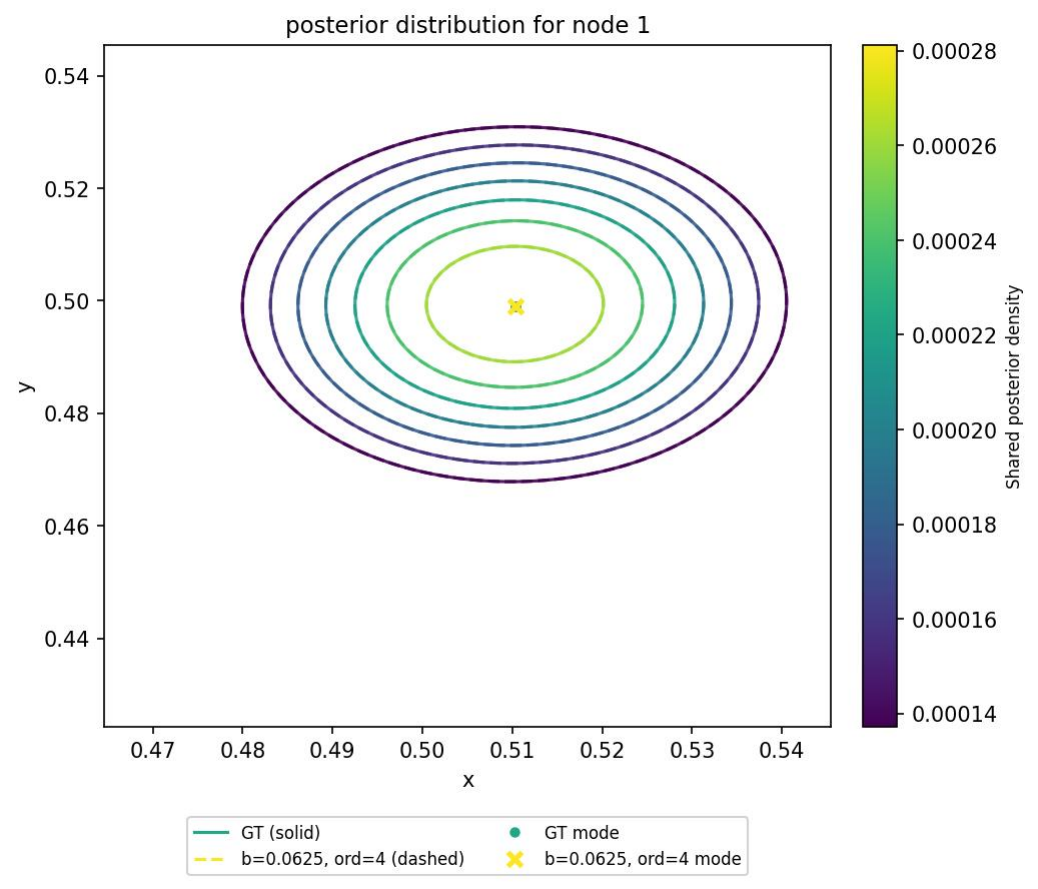}
        \caption{$b=0.0625$, Algorithm {\ref{alg_MCMC_full}}, {$\kappa=4$}}
        \label{fig:single_node_b00625_order4}
    \end{subfigure}

    \vspace{0.6em}

    \begin{subfigure}[b]{0.31\textwidth}
        \centering
        \includegraphics[width=\linewidth]{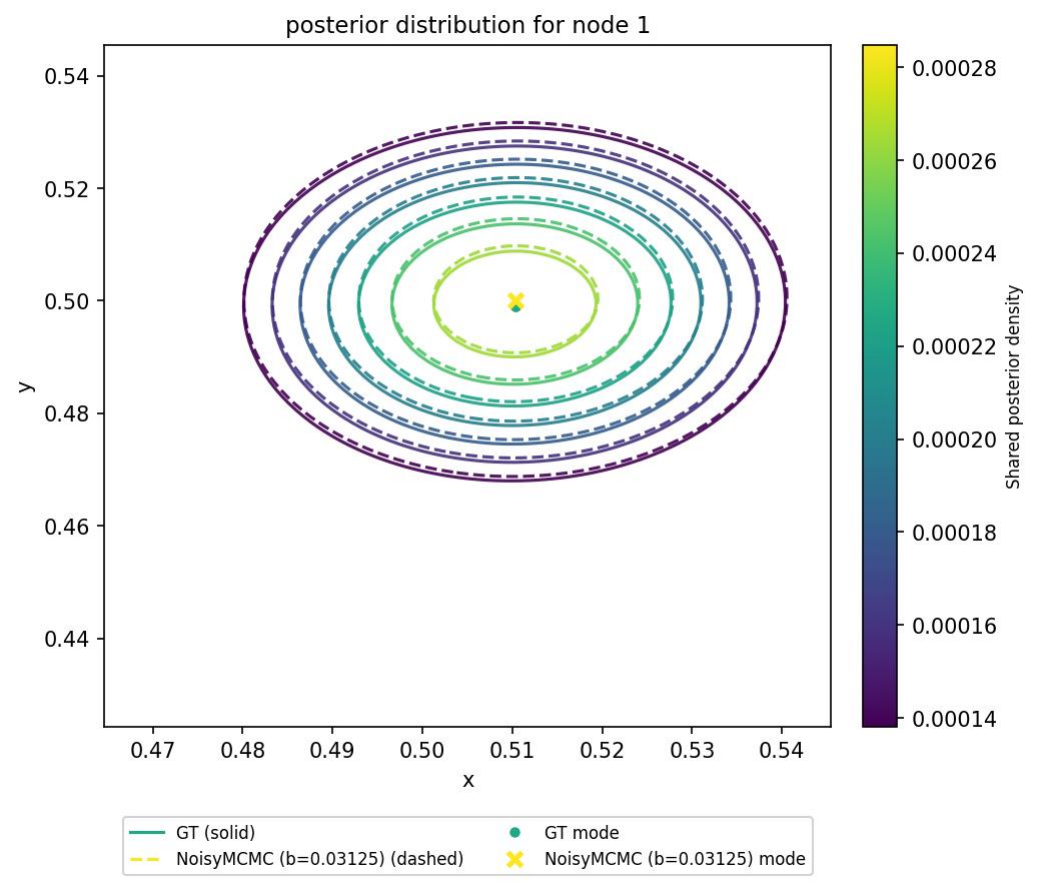}
        \caption{$b=0.03125$, Noisy MCMC}
        \label{fig:single_node_b003125_friel}
    \end{subfigure}
    \hfill
    \begin{subfigure}[b]{0.31\textwidth}
        \centering
        \includegraphics[width=\linewidth]{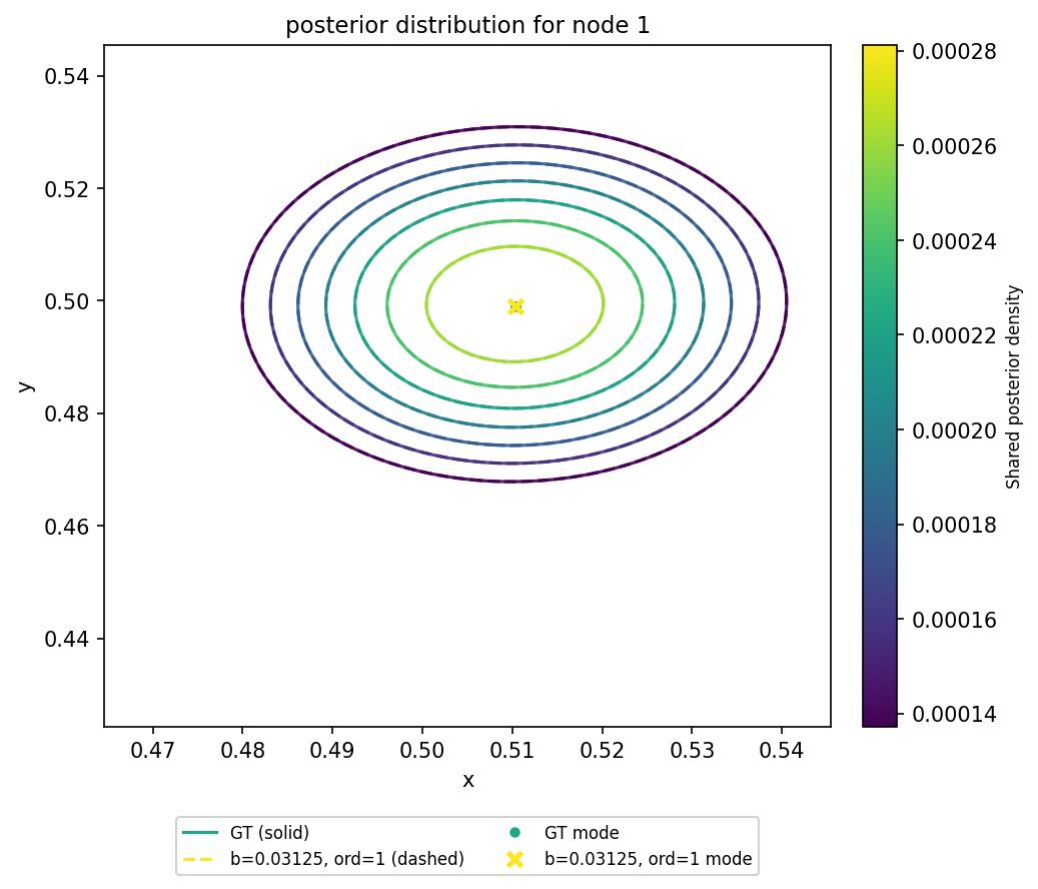}
        \caption{$b=0.03125$, Algorithm {\ref{alg_MCMC_full2}}, {$\kappa=1$}}
        \label{fig:single_node_b003125_order1}
    \end{subfigure}
    \hfill
    \begin{subfigure}[b]{0.31\textwidth}
        \centering
        \includegraphics[width=\linewidth]{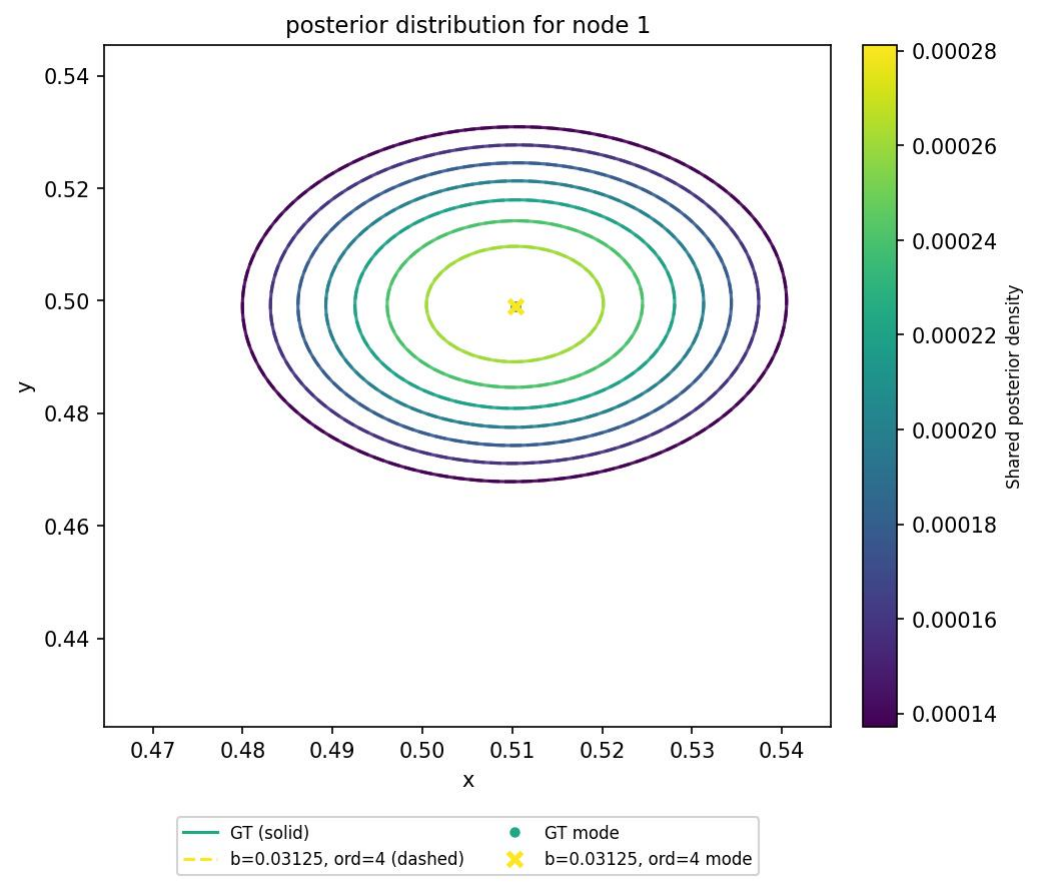}
        \caption{$b=0.03125$, Algorithm {\ref{alg_MCMC_full}}, {$\kappa=4$}}
        \label{fig:single_node_b003125_order4}
    \end{subfigure}

    \caption{Posterior distribution of node  under different block lengths and approximation methods}
    \label{fig:single_node_all}
\end{figure}
We use the same parameter setting as in Section \ref{SubsecRainbow} and simulate a network with $N=2000$ nodes, yielding a graph density close to $50\%$. We compare three methods to the ground truth MwG: the grid-partition noisy MCMC of \cite{rastelli2018computationally}, and our moments-based approximation in Algorithm~\ref{alg_MCMC_full} and Algorithm~\ref{alg_MCMC_full2}.\par
For a meaningful comparison, we select a node whose true latent position lies close to the center of the latent space, so that its posterior is not dominated by boundary effects. 

For this selected node, we evaluate the marginal posterior on a two-dimensional grid and visualize the resulting posterior contour under different approximation settings. We consider several block lengths $b \in \{0.125, 0.0625, 0.03125\}$, and compare first-order and fourth-order moments-based approximations against the noisy MCMC benchmark.

Figure~\ref{fig:single_node_all}  shows the posterior contours for the selected node with different methods, and there is a clear dependence on both the block length $b$ and the Taylor order $\kappa$. For the noisy MCMC benchmark, the posterior contours are broadly similar to the ground truth, although some discrepancies are visible when the block length is large. For the moments-based approximation, the first-order method with $\kappa=1$ can produce substantial distortion when $b$ is large, including shifts in the posterior mode and deformation of the contour shape. As $b$ decreases, the first-order approximation improves, and for $b=0.03125$ the fitted posterior becomes much closer to the reference distribution.

In contrast, the fourth-order approximation with $\kappa=4$ is stable across all tested block lengths. Its contours align closely with the ground-truth posterior even for relatively coarse partitions. This indicates that increasing the Taylor order substantially reduces sensitivity to the block length and improves local posterior accuracy.

\subsection{Distance PDF}\label{dist_pdf}

 In the third study, we assess whether the approximate sampler preserves the distribution of pairwise distances in the latent space. This is crucial because many inferential features of latent position models depend more directly on inter-point distances than on absolute coordinates.

Using posterior samples from both the ground-truth and approximate samplers, we compute the empirical distribution of distances between node (the one closet to the center of space) and three other nodes chosen from different structural regimes: a nearby node, an intermediate-distance node, and a well-separated node. This allows us to evaluate the approximation across both local and global distance scales.
\begin{figure}[htbp]
    \centering
    \begin{subfigure}[b]{0.32\linewidth}
        \centering
        \includegraphics[width=\linewidth]{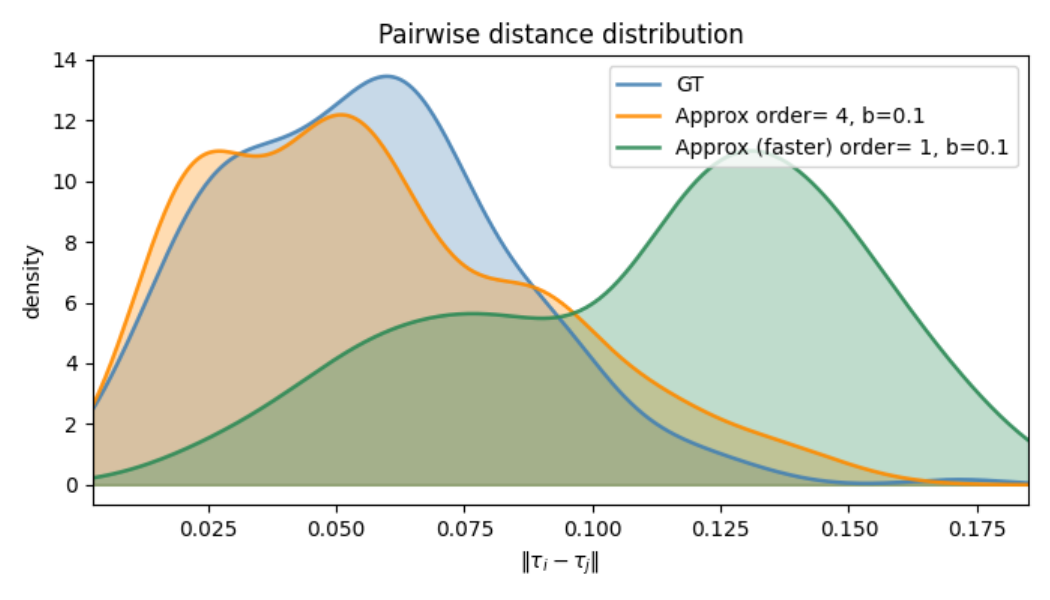}
        \caption{nearby nodes}
        \label{fig:small_graph_distribution_near}
    \end{subfigure}
    \hfill
    \begin{subfigure}[b]{0.32\linewidth}
        \centering
        \includegraphics[width=\linewidth]{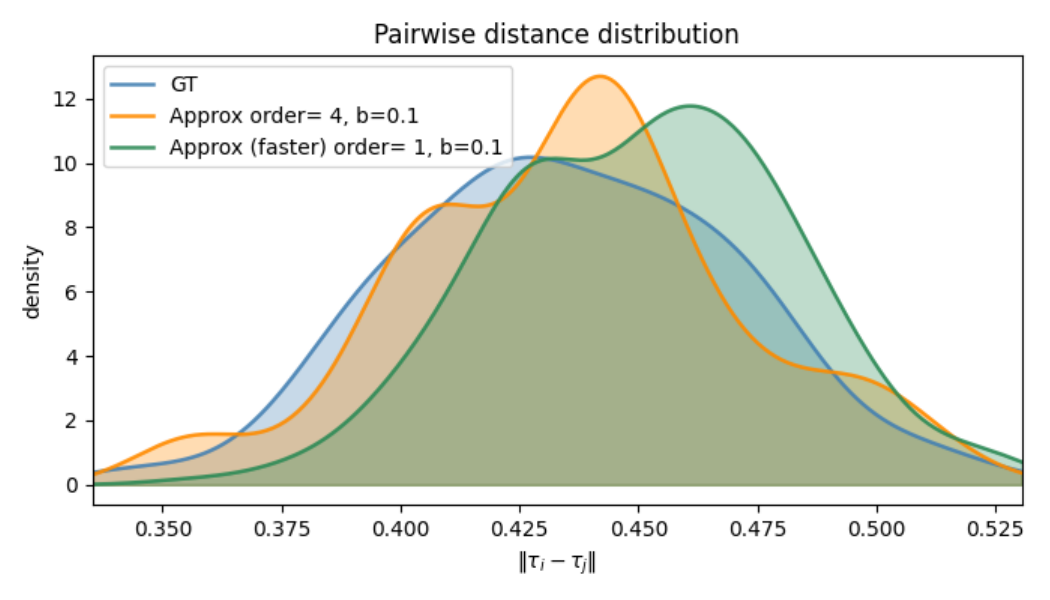}
        \caption{intermediate distance}
        \label{fig:small_graph_distribution_intermediate}
    \end{subfigure}
    \hfill
    \begin{subfigure}[b]{0.32\linewidth}
        \centering
        \includegraphics[width=\linewidth]{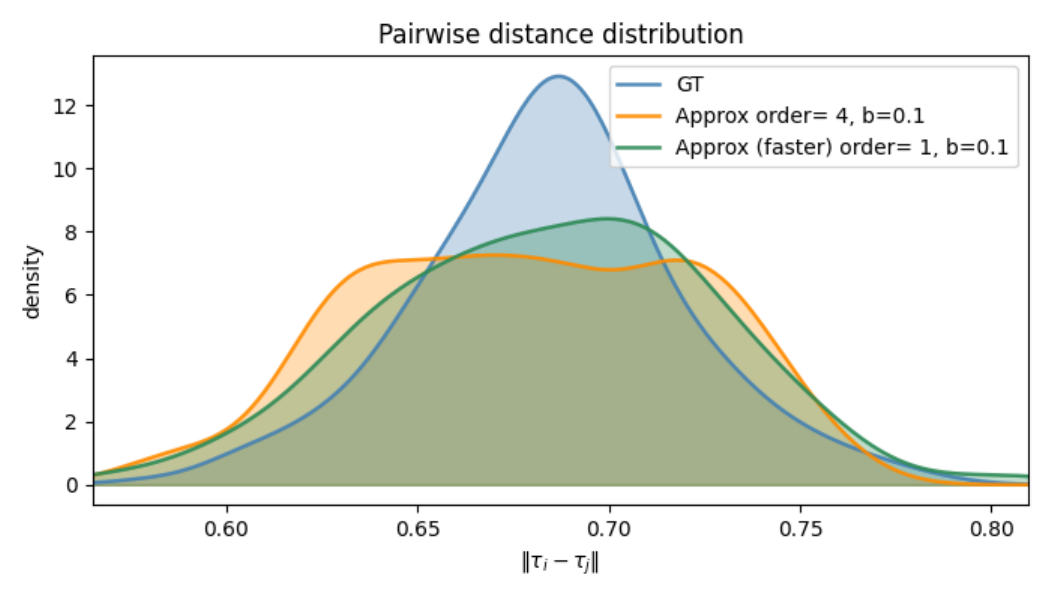}
        \caption{well-separated nodes}
        \label{fig:small_graph_distribution_far}
    \end{subfigure}
    \caption{Empirical PDFs of posterior pairwise distances for selected node pairs under the ground-truth and approximate samplers with nodes $N=2000$.}
    \label{fig:distance_pdf_small}
\end{figure}
\begin{figure}[htbp]
    \centering
    \begin{subfigure}[b]{0.32\linewidth}
        \centering
        \includegraphics[width=\linewidth]{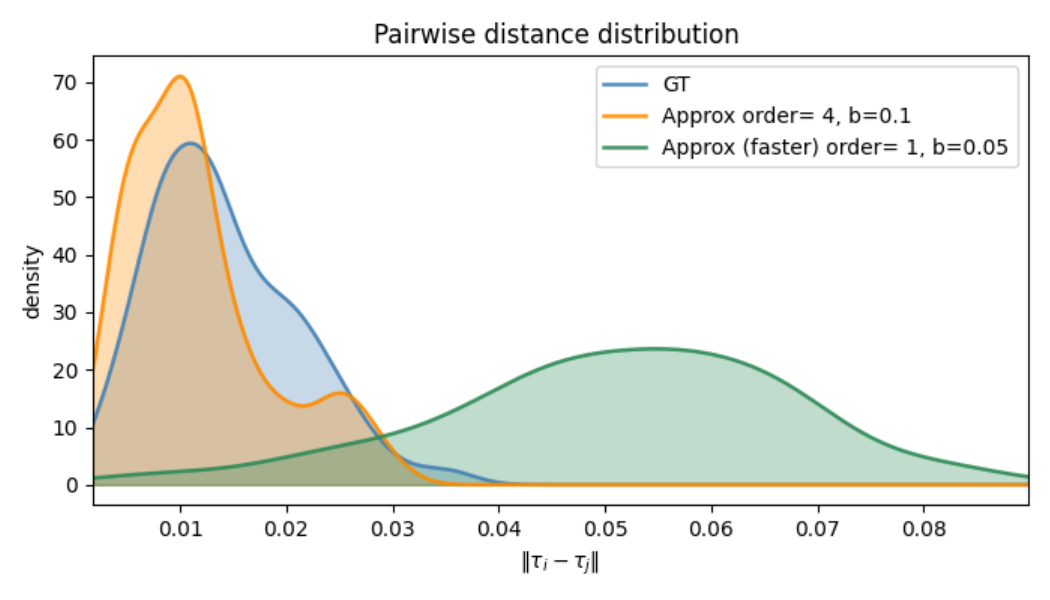}
        \caption{nearby nodes}
        \label{fig:large_graph_distribution_near}
    \end{subfigure}
    \hfill
    \begin{subfigure}[b]{0.32\linewidth}
        \centering
        \includegraphics[width=\linewidth]{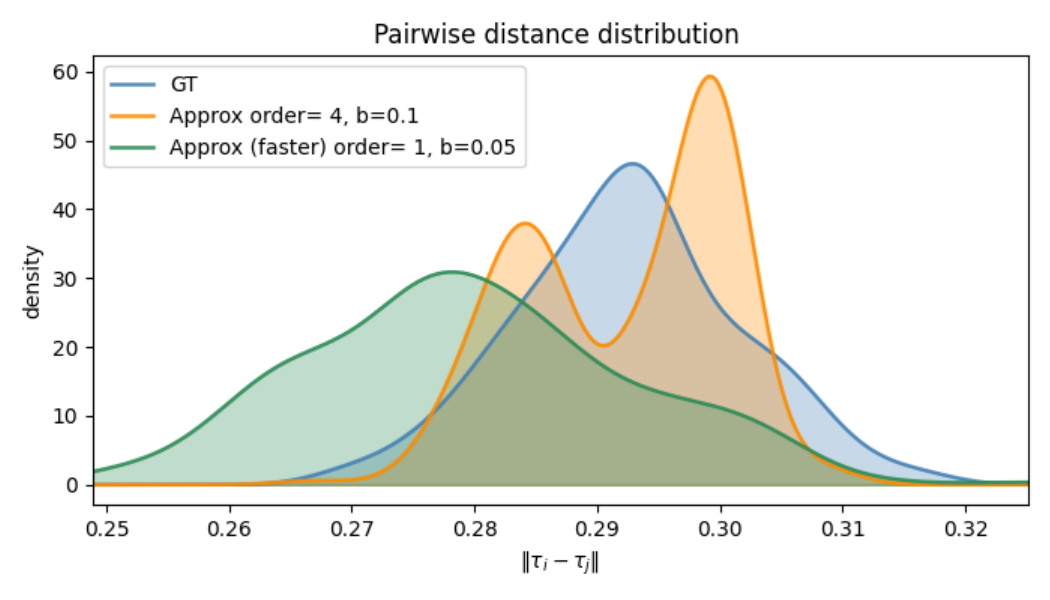}
        \caption{intermediate distance}
        \label{fig:large_graph_distribution_intermediate}
    \end{subfigure}
    \hfill
    \begin{subfigure}[b]{0.32\linewidth}
        \centering
        \includegraphics[width=\linewidth]{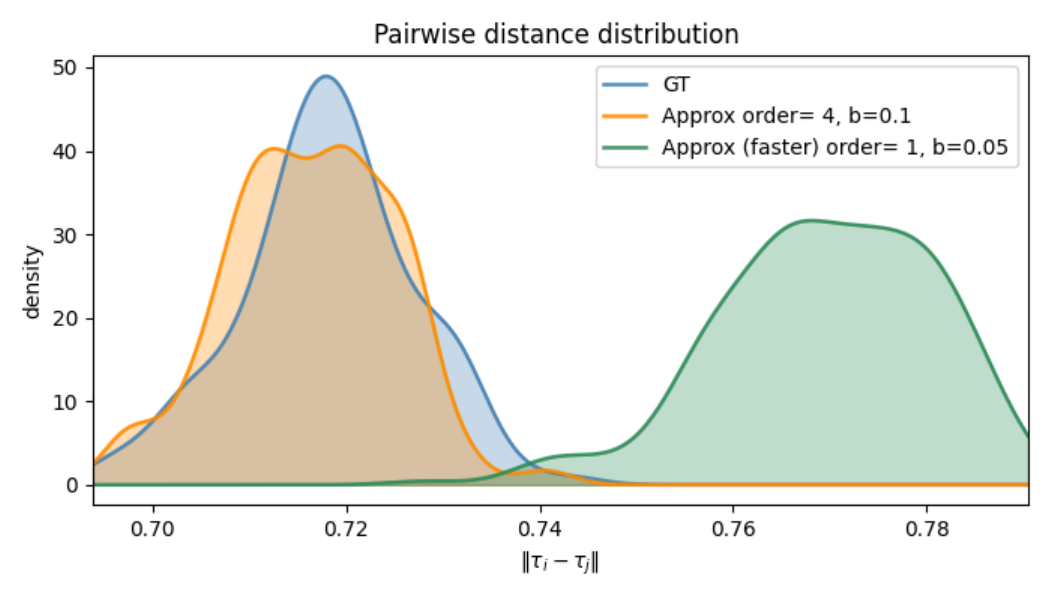}
        \caption{well-separated nodes}
        \label{fig:large_graph_distribution_far}
    \end{subfigure}
    \caption{Empirical PDFs of posterior pairwise distances for selected node pairs under the ground-truth and approximate samplers with nodes $N=30000$.}
    \label{fig:distance_pdf_large}
\end{figure}

Figures~\ref{fig:distance_pdf_small} and~\ref{fig:distance_pdf_large} evaluate whether the approximate samplers preserve posterior uncertainty in pairwise latent distances. 
Here fourth-order approximation generally matches the ground-truth pairwise-distance distributions well, both in terms of modal location and overall density shape. The agreement is especially strong for intermediate and well-separated pairs, where the posterior distance distributions are relatively stable. The first-order approximation captures the rough scale of the distances, but shows more noticeable bias and spread. These discrepancies are more pronounced in the $N=30000 $ experiment, suggesting that the first-order approximation may become less reliable for fine-scale distance uncertainty in larger graphs.
\paragraph{\textbf{Order statistics of nearest-neighbor distances}}
As a supplementary sanity check on the local neighborhood structure, we examine the order statistics of nearest-neighbor distances for the central node. We sort the latent distances from the selected node to all other nodes and compare the resulting empirical distributions between the ground-truth and approximate samplers.

\begin{figure}[htbp]
    \centering
    \begin{subfigure}[b]{0.49\linewidth}
        \centering
        \includegraphics[width=\linewidth]{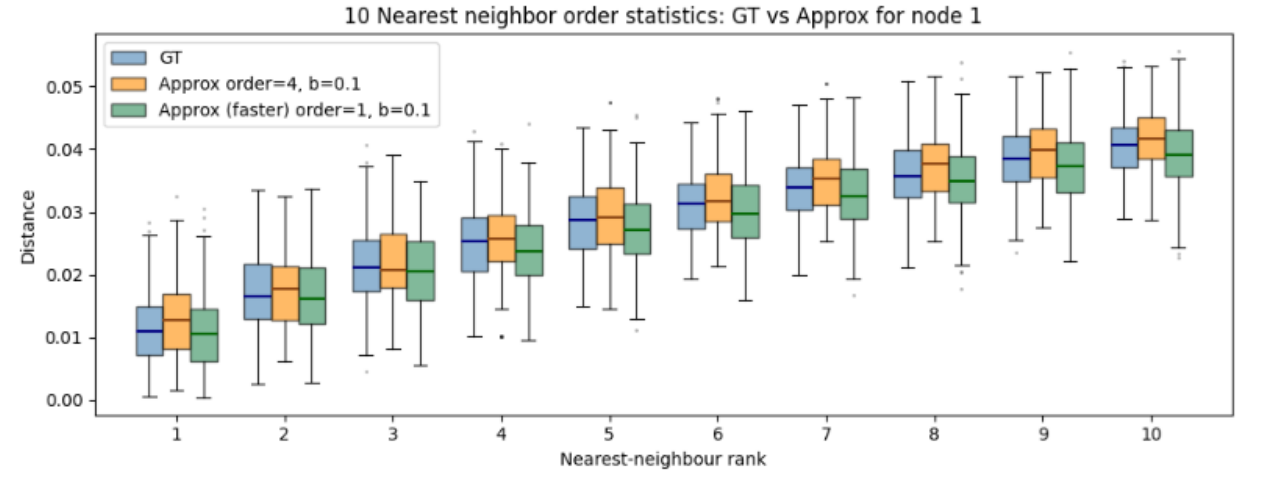}
        \caption{small graph $N=2000$}
        \label{fig:order_stat_small}
    \end{subfigure}
    \hfill
    \begin{subfigure}[b]{0.49\linewidth}
        \centering
       \includegraphics[width=\linewidth]{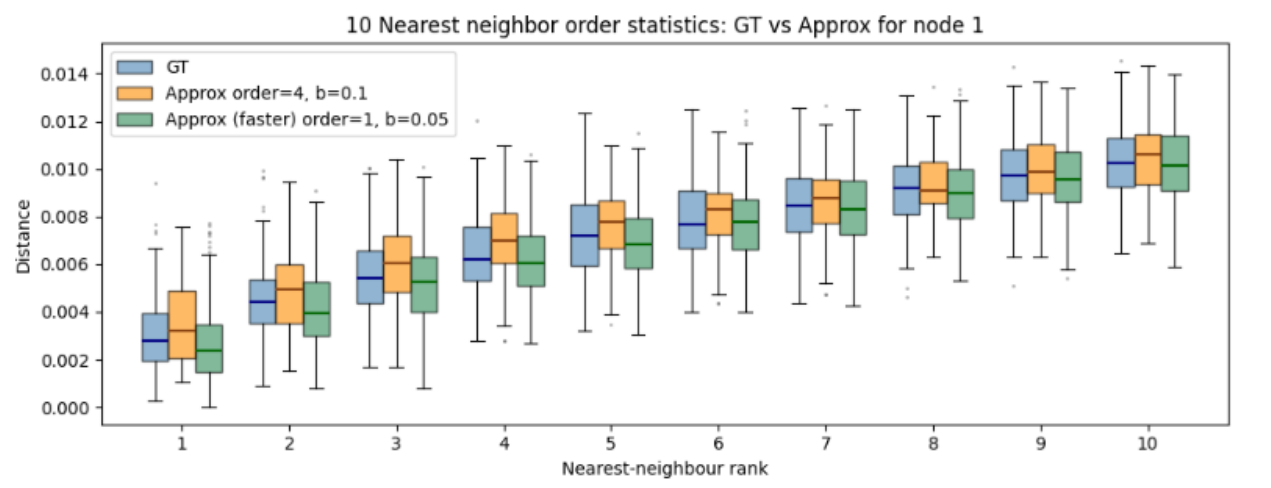}
        \caption{large graph $N=30000$}
        \label{fig:order_stat_large}
    \end{subfigure}
    
    \caption{Order statistics of nearest-neighbor distances for selected node under the ground-truth and approximate samplers}
    \label{fig:order_stat}
\end{figure}

 In Figure~\ref{fig:order_stat}, the fourth-order approximation closely follows this pattern for both $N=2000$ and $N=30000$, matching the ground truth in median, spread, and monotone rank-wise trend. This indicates that the approximation preserves the density of small pairwise distances and hence the local neighborhood geometry. The first-order approximation captures the qualitative increasing trend, but exhibits larger shifts and variability. Especially in the small graph, it shows large differences in density and range compared with the ground truth. For the large graph, even though we take a smaller partition with $b=0.05$ for $\kappa=1$ (while using $b=0.1$ is hard to capture the global structure), the first order expansion still differs in the rank $1-3$, then improves for larger ranks. These deviations are consistent with the discrepancies observed in the pairwise-distance PDFs in Figures~\ref{fig:distance_pdf_small} and~\ref{fig:distance_pdf_large}.
\section{Conclusions} \label{sec-conc}

We have presented an MCMC algorithm for fitting LPMs, and proved that it is both fast and accurate under moderately strong assumptions that provably hold for certain popular models. We mention some natural open questions left by this work.

\subsection{Application to Other Models}

Our general strategy of speeding up an MCMC algorithm by roughly chopping up its state space is quite popular (see \textit{e.g.} the survey \cite{rudolf2024perturbationsmarkovchains}). This paper proposes a small twist: rather than keeping track of the approximate location of observations within this partition, instead keep track of some nearly-sufficient statistics associated with the partition. This general idea has been used with some success for simple regression-like models (see \textit{e.g. } \cite{huggpass}), and we believe it can be applied fruitfully to more complex dependence structures. While we have focused on Taylor expansions, we suspect that other expansions (such as the Legendre polynomials commonly used in Gaussian quadrature) will be more appropriate in other contexts.

\subsection{Applications to Complicated LPMs} \label{RemMulti}

In the interest of keeping the presentation of our algorithms simple, we've restricted our attention to rather simple LPMs: we don't allow for node features, we assume that all embeddings with high posterior probability are close to (some isometry of) a single good embedding, and so on. These simple LPMs appear to be very popular in practice, but of course not all LPM-like models have these properties.

While we don't think it is possible to fit \textit{all} LPMs efficiently using the methods in our paper, we do think it is possible to substantially extend these methods. Consider the case of ``slightly" multimodal posteriors, where the posterior can be written as a mixture of a small number of unimodal distributions $\pi^{(1)}, \ldots \pi^{(k)}$. In this case, if you can find partitions $a_{1}, \ldots, a_{k}$ that are good (in the sense of Assumption \ref{assum3}) for each mixture component, then the intersection $a$ of these partitions is good (again in the sense of Assumption \ref{assum3}) for the posterior. As long as $k$ does not grow with $n$, taking intersections in this way does not affect any of our theorems about asymptotic complexity or error; a $k$ that grows slowly would have only a small impact. Of course, ``severe" multimodality is possible, but this sort of ``gentle" multimodality seems to be quite important \cite{Priebe_2019}. 

Even if it is difficult to find a small fixed number of partitions $a_{1},\ldots,a_{k}$ that can be intersected to satisfy  Assumption \ref{assum3}, it may be possible to use our algorithms. The simplest approach here is to simply re-initialize the algorithm occasionally, finding a partition that is \textit{locally} appropriate.

\bibliographystyle{alpha}
\bibliography{ref}

\newpage 

\appendix 

\section{Node Partition Construction}  \label{SubsecNodePart}

Fix a graph $G = (V,E)$ with $n$ vertices. Let $\tau \in (\mathbb{R}^{2})^{n}$ be a single embedding of the vertices of this graph. In this section, we construct a simple partition $a$ of $V$ based on $\tau$ and show that it has certain good properties as long as $\tau$ is sufficiently accurate - see Lemma \ref{IneqBasicPartitionProperties}. The construction in this section works for any $\tau$, but we think of it as coming from a point estimator with good theoretical properties such as the spectral embedding \cite{von2008consistency}. 

We begin by defining a ``system of boxes." Fix a desired accuracy $b > 0$. For $g,h \in \mathbb{R}$, define $B[g,h]$ to be the square with corners located at $(bg-b,bh-b)$, $(bg-b,bh)$, $(bg,bh)$ and $(bg,bh-b)$. Then define $C[g,h]=(bg-\frac{b}{2},bh-\frac{b}{2})$, the center of $B[g,h]$. We fix a set of special coordinates:
    \be \label{def_T}
    \mathbb{T} = \left\{\frac{z}{3} \, : \, z \in \mathbb{Z} \right\}
    \ee 
and finally the collection of overlapping boxes $\mathfrak{B} = \{B[g,h]\}_{g,h \in \mathbb{T}}$.

We next fix an assignment function $\mathcal{A} \, : \, [n] \to \mathbb{T}^{2}$ by: 
\be \label{assign}
\mathcal{A}(i) = \text{argmin}_{ (g,h) \in \mathbb{T}^{2}} \| \tau_{i} - C[g,h] \|_{\infty},
\ee 

breaking ties arbitrarily.

This construction is illustrated in Figure \ref{OverlappingBoxes3}. The main idea is that no node $i$ is close to the boundary of its assigned box $B[\mathcal{A}(i)]$.

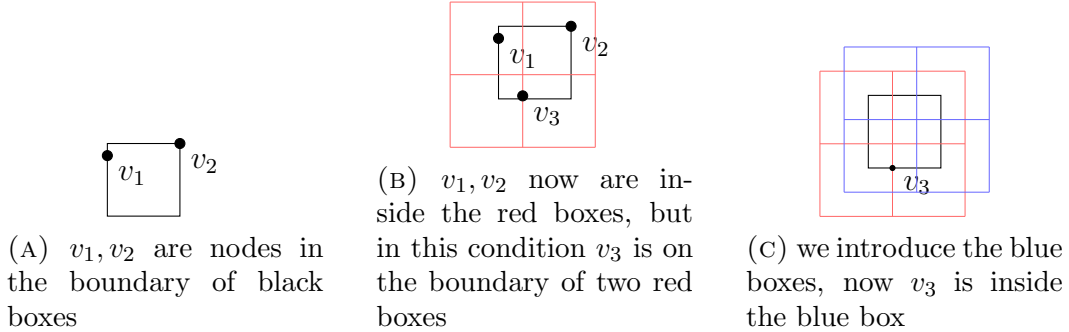
\begin{figure}[H]
    \caption{Example of overlapping boxes system $\mathfrak{B}$}
    \label{OverlappingBoxes3}
    \centering
    \begin{subfigure}[b]{0.3\textwidth}
        \centering
        \begin{tikzpicture}[scale=0.4]
            \draw (0,0) rectangle (2.4,2.4);
            \coordinate (a) at (0,2);
            \draw[fill=black] (0,2) circle(0.4em);
            \node[below right] at (0,2) {$v_1$};
            \coordinate (b) at (2.4,2.4);
            \draw[fill=black] (2.4,2.4) circle(0.4em);
            \node[below right] at (2.4,2.4) {$v_2$};
        \end{tikzpicture}
        \caption{$v_1, v_2$ are nodes in the boundary of black boxes}
    \end{subfigure}
    \hfill
    \begin{subfigure}[b]{0.3\textwidth}
        \centering
        {\begin{tikzpicture}[scale=0.4]
            \draw (0,0) rectangle (2.4,2.4);
            \coordinate (a) at (0,2);
            \draw[fill=black] (0,2) circle(0.4em);
            \node[below right](a) at (0,2) {$v_1$};
            \coordinate (b) at (2.4,2.4);
            \draw[fill=black] (2.4,2.4) circle(0.4em);
            \node[below right](b) at (2.4,2.4) {$v_2$};
            \draw[help lines, color=red!60] (-1.6,0.8)--(3.2,0.8) ;
            \draw[help lines, color=red!60] (-1.6,3.2)--(3.2,3.2) ;
            \draw[help lines, color=red!60] (0.8,-1.6)--(0.8,3.2) ;
            \draw[help lines, color=red!60] (3.2,-1.6)--(3.2,3.2) ;
            \draw[help lines, color=red!60] (-1.6,-1.6)--(-1.6,3.2) ;
            \draw[help lines, color=red!60] (-1.6,-1.6)--(3.2,-1.6) ;
            \coordinate (c) at (0.8,0.1);
            \draw[fill=black] (0.8,0.1) circle(0.4em);
            \node[below right](c) at (0.8,0.1) {$v_3$};
        \end{tikzpicture}}
        \caption{$v_1,v_2$ now are inside the red boxes, but in this condition $v_3$ is on the boundary of two red boxes}
    \end{subfigure}
    \hfill
    \begin{subfigure}[b]{0.3\textwidth}
        \centering
        \begin{tikzpicture}[scale=0.4]
            \draw (0,0) rectangle (2.4,2.4);
            \draw[help lines, color=red!60] (-1.6,0.8)--(3.2,0.8) ;
            \draw[help lines, color=red!60] (-1.6,3.2)--(3.2,3.2) ;
            \draw[help lines, color=red!60] (0.8,-1.6)--(0.8,3.2) ;
            \draw[help lines, color=red!60] (3.2,-1.6)--(3.2,3.2) ;
            \draw[help lines, color=red!60] (-1.6,-1.6)--(-1.6,3.2) ;
            \draw[help lines, color=red!60] (-1.6,-1.6)--(3.2,-1.6) ;
            \coordinate (c) at (0.8,0);
            \draw[fill=black] (0.8,0) circle(0.2em);
            \node[below right](c) at (0.8,0.1) {$v_3$};
            \draw[help lines, color=blue!60] (-0.8,1.6)--(4,1.6) ;
            \draw[help lines, color=blue!60] (-0.8,4)--(4,4) ;
            \draw[help lines, color=blue!60] (1.6,-0.8)--(1.6,4) ;
            \draw[help lines, color=blue!60] (4,-0.8)--(4,4) ;
            \draw[help lines, color=blue!60] (-0.8,-0.8)--(-0.8,4) ;
            \draw[help lines, color=blue!60] (-0.8,-0.8)--(4,-0.8) ;
        \end{tikzpicture}
        \caption{we introduce the blue boxes, now $v_3$ is inside the blue box}
    \end{subfigure}
\end{figure}

After the assignment, we discard boxes with no assigned vertices. More precisely, define
\be \label{def_box}
\mathcal{B} = \{ B[g,h] \, : \, \exists i \in [n] \, \text{s.t.} \, \mathcal{A}(i) = (g,h)\}
\ee

To simplify notation, we write this list in the arbitrary order $\mathcal{B}=\{B_1,B_2,\dots, B_{m}\}$. Finally, define the function $a \, : \, [n] \to [m]$ to be the function mapping a node $i$ to the index of the box $B[\mathcal{A}(i)]$. 

We have the following basic guarantee:

\begin{lemma} \label{IneqBasicPartitionProperties}
Fix notation as in this section. Then 
\be 
\max_{i, j \, : \, a(i) = a(j)} \| \tau_{i} - \tau_{j} \| \leq \frac{\sqrt{2}}{3} b.
\ee
\end{lemma}
\begin{proof}
By inspection of Equation \eqref{assign}, we have for all vertices $i$ that 
\be 
\| \tau_{i} - C[\mathcal{A}(i)] \| \leq \frac{b}{3 \sqrt{2}}.
\ee 
Thus, by the triangle inequality,
\be 
\max_{ j \, : \, a(i) = a(j)} \| \tau_{i} - \tau_{j} \| \leq 2 \frac{b}{3 \sqrt{2}}.
\ee 

\end{proof}

\section{Proofs of Error Guarantees}

We prove the estimates in Section \ref{SecMainErrorBounds}. We begin with some generic preliminary bounds (Section \ref{SecPrelimRes}), then prove the deterministic likelihood bound in Proposition \ref{DetTaylorBound}, Theorem \ref{Thm1}, and finally Corollary \ref{CorrMain} (Section \ref{proof1}).

\subsection{Preliminary results} \label{SecPrelimRes}

We collect generic estimates that are useful in proving our main results.

\begin{lemma}\label{Thm_1_lemma_4}
Fix $\eta > 0$ and two probability densities $f,g$ with respect to a common measure. If the pointwise multiplicative bounds
\[
\frac{g}{f}\leq 1+\eta, \qquad \frac{f}{g}\leq 1+\eta
\]
hold wherever the ratios are defined, then
\be\label{tv_bound_lemma}
\|f-g\|_{TV}\leq \eta.
\ee
\end{lemma}

\begin{proof}
Let $A = \{x \, : \, f(x) > g(x)\}$. Using the standard convention $\|f-g\|_{TV}=\int_A(f-g)$,
\be
\|f-g\|_{TV}=\int_A(f-g)\leq \int_A \eta g \leq \eta.
\ee
\end{proof}







We restate the well-known bound:

\begin{lemma}\label{Thm_1_lemma_3}
Fix $0 < C_{1} < C_{2} < \infty$ and two nonnegative functions $f,g$ on the same domain $\mathcal{X}$ that satisfy $0<C_1\leq\frac{f(x)}{g(x)}\leq C_2$ for all $x \in \mathcal{X}$. Then for any measure $\mu$ on $\mathcal{X}$ for which $0<\int g\,d\mu<\infty$,
\be\label{bound_for_sequence}
C_1\leq\frac{\int f d \mu}{\int g d \mu}\leq C_2.
\ee
\end{lemma}

\begin{proof}
The pointwise inequality implies $C_1 g(x)\leq f(x)\leq C_2 g(x)$ for all $x\in\mathcal{X}$. Integrating with respect to $\mu$ gives
\[
C_1\int g\,d\mu \leq \int f\,d\mu \leq C_2\int g\,d\mu.
\]
Dividing by the positive quantity $\int g\,d\mu$ gives \eqref{bound_for_sequence}.
\end{proof}

\subsection{Proofs of the Deterministic Bound, Theorem \ref{Thm1}, and the Random-Graph Corollary} \label{proof1}

\begingroup
\subsubsection{Likelihood estimate}

We prove the deterministic likelihood bound. Recall from Equations \eqref{EqLogLikSimple} and \eqref{approx_log_like} that $\tilde{L}_{\theta}(\tau)$ is obtained by replacing each term $g_{\ell,\theta}(\tau_i,\tau_j)$ with its $\kappa$-th order Taylor expansion around $(y_{a(i)},y_{a(j)})$ and then collecting the corresponding moments. If $\tau\in S_{a,b}$, then each $y_s(\tau)$ lies in the convex hull of $\{\tau_i:a(i)=s\}\subset[-\mathcal{H},\mathcal{H}]^2$. Thus the centers and the line segments used in the Taylor expansion lie in the neighborhood $\mathcal{U}$ from Assumption \ref{assum2}. Moreover, $\|\tau_i-y_{a(i)}\|_{2} \leq  b$ for every $i$.

\begingroup
Set $m=\kappa+1$ and $h=((\tau_i,\tau_j)-(y_{a(i)},y_{a(j)}))\in\mathbb{R}^4$. Since $\tau\in S_{a,b}$, $\|h\|_1\leq 2 \sqrt{2} b$. The multivariate Taylor theorem and \eqref{assum4} give 
\be
D_{m,g}\sum_{|\alpha|=m}\frac{|h^\alpha|}{\alpha!}
= D_{m,g}\frac{\|h\|_1^m}{m!}
\leq M_g\left(\frac{2 \sqrt{2} b}{\rho_g}\right)^m
= M_g(B_g b)^{\kappa+1}.
\ee
There are $n(n-1)$ directed pairs and the log-likelihood has the prefactor $1/2$, so
\be \label{EqSImpleTaylorLikBound}
|\tilde{L}_{\theta}(\tau)-L_{\theta}(\tau)| \leq \frac{n(n-1)}{2} M_g(B_g b)^{\kappa + 1}=R_{n,\kappa}(b).
\ee 
Exponentiating this additive log-likelihood bound gives \eqref{boundlike_mult}.
\endgroup

\subsubsection{TV distance}

We prove Theorem \ref{Thm1}. Let $S=S_{a,b}\times\Theta$ and $R=R_{n,\kappa}(b)$. Denote by $\pi_{G}^{S}$ the exact posterior $\pi_G$ conditioned on $S$. The unnormalized exact and approximate posterior densities on $S$ are proportional to
\be
h(\tau,\theta)=\exp(L_{\theta}(\tau))\pi(\tau,\theta),\qquad \hat h(\tau,\theta)=\exp(\tilde L_{\theta}(\tau))\pi(\tau,\theta).
\ee
By \eqref{boundlike_mult}, $e^{-R}\leq \frac{\hat h}{h} ,\frac{h}{\hat{h}}\leq e^{R}$ on $S$. Lemma \ref{Thm_1_lemma_3} gives the same multiplicative bounds for the ratio of the corresponding normalizing constants. Therefore the normalized densities of $\pi_G^S$ and $\hat\pi_G$ satisfy
\be\label{density_ration_bound_1}
e^{-2R}\leq \frac{d\hat\pi_G}{d\pi_G^S}\leq e^{2R}, \qquad e^{-2R}\leq \frac{d\pi_G^S}{d\hat\pi_G}\leq e^{2R}.
\ee
Applying Lemma \ref{Thm_1_lemma_4} with $\eta=e^{2R}-1$ yields
\be
\|\pi_G^S-\hat\pi_G\|_{TV}\leq e^{2R}-1.
\ee
By assumption, $\pi_G(S)\geq 1-\epsilon_{\mathrm{post}}$, so
\be
\|\pi_G-\pi_G^S\|_{TV}=1-\pi_G(S)\leq \epsilon_{\mathrm{post}}.
\ee
The triangle inequality gives
\be
\|\pi_G-\hat\pi_G\|_{TV}\leq e^{2R_{n,\kappa}(b)}-1+\epsilon_{\mathrm{post}}.
\ee

\endgroup

\end{document}